\documentclass[aps,showpacs,prl,twocolumn,superscriptaddress]{revtex4-1}

\usepackage{exscale}
\usepackage{bbm}
\usepackage{graphicx}
\usepackage{amsmath}
\usepackage{latexsym}
\usepackage{amsfonts}
\usepackage{amssymb}
\usepackage{times}
\usepackage[T1]{fontenc}
\usepackage{amsthm}
\usepackage{enumerate}
\usepackage{bbold}
\usepackage{color}
\usepackage[colorlinks=true,citecolor=blue,urlcolor=blue]{hyperref}
\usepackage{mathtools}
\usepackage{changes}
\usepackage[normalem]{ulem}
\usepackage{amsfonts,amsmath,amssymb,amsthm}

\newcommand{\ket}[1]{\left|#1\right\rangle}

\newcommand{\tp}{\otimes}
\theoremstyle{plain}

\newtheorem{fakt}{Fact}
\newtheorem*{defin}{Definition}

\newcommand{\tor}{\mathrm{tor}}
\newcommand{\sz}{\mathrm{Z}}
\newcommand{\sx}{\mathrm{X}}
\newcommand{\sy}{\mathrm{Y}}
\newcommand{\rA}{\mathrm{A}}

\begin{document}

\title{Device-Independent Certification of Genuinely Entangled Subspaces}
\author{Flavio Baccari}
\email{flavio.baccari@mpq.mpg.de}
\affiliation{Max-Planck-Institut f\"ur Quantenoptik, Hans-Kopfermann-Stra{\ss}e 1, 85748 Garching, Germany}
\author{Remigiusz Augusiak}
\email{augusiak@cft.edu.pl}
\affiliation{Center for Theoretical Physics, Polish Academy of Sciences, Aleja Lotnik\'{o}w 32/46, 02-668 Warsaw, Poland}
\author{Ivan \v{S}upi\'{c}}
    \affiliation{D\'{e}partement de Physique Appliqu\'{e}e, Universit\'{e} de Gen\`{e}ve, 1211 Gen\`{e}ve, Switzerland}
\author{Antonio Ac\'{i}n}
\affiliation{ICFO-Institut de Ciencies Fotoniques, The Barcelona Institute of Science and Technology, 08860 Castelldefels (Barcelona), Spain}
\affiliation{ICREA - Instituci\'{o} Catalana de Recerca i Estudis Avancats, 08011 Barcelona, Spain}

\begin{abstract}
    Self-testing is a procedure for characterising quantum resources with the minimal level of trust. Up to now it has been used as a device-independent certification tool for particular quantum measurements, channels and pure entangled states. In this work we introduce the concept of self-testing more general entanglement structures. More precisely, we present the first self-tests of an entangled subspace - the five-qubit code and the toric code. We show that all quantum states maximally violating a suitably chosen Bell inequality must belong to the corresponding code subspace, which remarkably includes also mixed states.
\end{abstract}

\maketitle

\emph{Introduction.}--- 
Authentically quantum effects such as entanglement and measurement incompatibility play a key role in the development of various quantum information protocols. In this context, verifying that a device or an algorithm indeed uses quantum resources is a very important task. There are many frameworks for such kinds of verification, in a broad sense known as \emph{testing of quantum properties} \cite{Montanaro_Wolf}. In a standard quantum property testing scenario a user, usually called verifier, aims to certify that their devices, commonly named provers, exploit some quantum resource. 

The strongest form of verification is \emph{device-independent} (DI)  \cite{acin2007device,pironio2010random,colbeck} in which no assumptions are made on the devices; they are simply treated as black boxes. A DI certification performed by a completely classical verifier is also known as \emph{self-testing} \cite{Mayers2004,STreview}. In such a scenario the only way to verify the 'quantumness' of the provers is to interact with them, for example by asking them some questions by means of classical communication channels and receiving the answers through the same channels (see Fig. \ref{fig:DIcert}). Any information about the underlying physical system is then inferred by the verifier from the observed correlations between those answers.

\begin{figure}
\centering
\includegraphics[width=0.7\columnwidth]{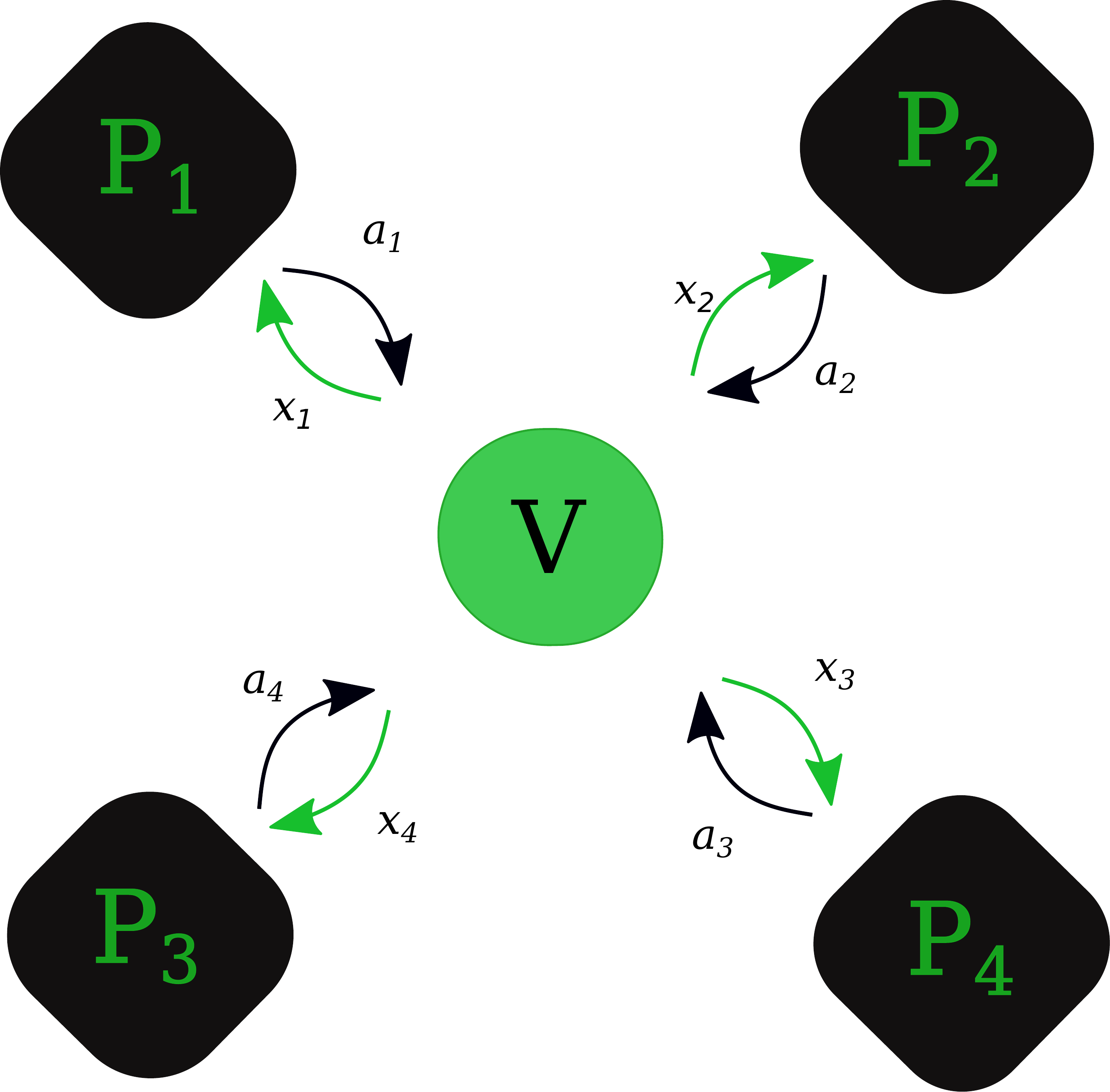}
\caption{\textbf{A model of self testing}: a classical verifier $V$ aims to certify a particular quantum property of the resource shared by non communicating provers $P_i$. The verifier sends different classical inputs $x_i$ to the provers and they respond with classical outputs $a_i$. If nonlocal, correlations between the outputs allow the verifier to make nontrivial conclusions about that quantum property.} 
\label{fig:DIcert}
\end{figure}

In a self-testing scenario, a central concept is that of Bell nonlocality \cite{bellreview}. Observing nonlocal behaviours is essential to certify several interesting properties of quantum systems, such as the exact form of a quantum state \cite{Coladangelo2017,McKague2014,Ivan}, measurement \cite{Jed1,wagnerPOVM} or a channel \cite{Magniez,Sekatski}, 
all this up to certain well-understood equivalences. However, 
self-testing has so far been deemed as a procedure tailored to a single quantum state and it has been a highly nontrivial question if it can be used to certify less specific quantum properties. 

Here we address this problem and introduce the definition of self-testing of an entangled subspace. Although Bell inequalities maximally violated by more than a single pure state are already known (e.g. Refs. \cite{AugusiakHorodecki2006,Goh2018}), we show for the first time that this kind of violation can be exploited to certify that a quantum state belongs to such a subspace. Hence, we present a relaxed definition of self-testing that is not able to distinguish between any mixture of the vectors belonging to the self-tested subspace. We present a first application of such a relaxed self-testing by constructing Bell inequalities whose maximal violation is attained by states belonging to some genuinely entangled subspaces, which are subspaces of multipartite Hilbert spaces consisting of only genuinely entangled states (see, e.g., \cite{DemianowiczAugusiak18}). As a paradigmatic example, we focus on those used in stabilizer error correcting codes \cite{GottesmanThesis}, such as the five-qubit code \cite{5qubitA,5qubitB} and the toric code \cite{Kitaev1997}, the latter allowing us to show that self-testing of subspaces is possible for systems composed of any number of particles. Interestingly, while still being based on the stabilizer formalism, our Bell inequalities are inequivalent to those presented in \cite{Srikanth2018,waegell2019benchmarks}.

\emph{Preliminaries.} We begin with some preliminaries. 

\emph{(1) The Bell scenario.} Following the scenario depicted in Fig. \ref{fig:DIcert}, let us consider $N$ spatially separated provers $P_i$ sharing a quantum state $\rho_P$ that acts on some finite-dimensional Hilbert space $\mathcal{H}_P = \mathcal{H}_{P_1}\tp\cdots\tp\mathcal{H}_{P_N}$ and a verifier $V$ asking them questions $x_i$. Upon receiving the question $x_i$ the prover $P_i$ measures a quantum observable $A_{x_i}^{(i)}$ on their share of $\rho_P$ and returns $V$ the outcome of that measurement $a_i$. Here we consider the simplest scenario in which all provers measure binary observables whose outcomes are labelled $a_i=\pm1$. If this procedure is repeated sufficiently many times, the verifier can estimate the vector $\mathcal{P}$ of expectation values 
\begin{equation}\label{correlations}
    \langle \rA_{x_{i_1}}^{(i_1)}\cdots  \rA_{x_{i_k}}^{(i_k)}\rangle = \textrm{Tr}\left[\left(\rA_{x_{i_1}}^{(i_1)}\otimes \cdots \otimes \rA_{x_{i_k}}^{(i_k)}\right)\rho_P \right]
\end{equation}
with $i_1<i_2<\ldots,<i_k$, $k=1,\ldots,N$ and $i_j=1,\ldots,N$.
Below we refer to $\mathcal{P}$ as behaviour or simply correlations.

The key ingredient making device-independent verification possible is that the correlations observed by the verifier exhibit quantum nonlocality \cite{bellreview}. The phenomenon of nonlocality consists of the existence of quantum correlations that cannot be reproduced by any local-realistic theory, or, phrasing alternatively, that violate Bell inequalities which bound the strength of 
correlations achievable in such theories. The most general form of a multipartite Bell inequality is
\begin{equation}
I_N := \sum_{k=1}^N 
\sum_{\substack{1 \leq i_1<i_2<\ldots<i_k \leq N  \\  x_{i_1},\ldots,x_{i_N}=0,1 }}\alpha^{i_1,\ldots,i_k}_{x_{i_1},\ldots,x_{i_k}}\langle \rA_{x_{i_1}}^{(i_1)} \ldots \rA_{x_{i_k}}^{(i_k)}\rangle \leq \beta_c,
\end{equation}
where $\beta_c$ is the local bound, that is, the maximal value of $I_N$ over all local-realistic correlations.

\emph{(2) Genuinely entangled subspaces.} Consider a bipartition of $N$ provers into two disjoint and nonempty groups $G$ and $G'$ and a pure multipartite state $\ket{\psi}\in\mathcal{H}_P$. We call it genuinely entangled if for any such bipartition $G|G'$ it cannot be written as a tensor product of pure states  
corresponding to $G$ and $G'$. We then call a subspace of $\mathcal{H}_P$ genuinely entangled if it consists of only genuinely multipartite entangled states (cf. Ref. \cite{DemianowiczAugusiak18}).

\emph{(3) Stabilizer codes.}  The $N$-fold tensor products of the Pauli operators $\{\sx,\sz,\sy,\openone\}$ with the overall factor $\pm 1$ or $\pm i$  forms the Pauli group $\mathbb{P}_N$ under matrix multiplication. A stabilizer $\mathbb{S}_N$ is any Abelian subgroup of $\mathbb{P}_N$, and the stabilizer coding space $C_N$ consists of all vectors $\ket{\psi}$ such that $S_i\ket{\psi} = \ket{\psi}$ for any $S_i \in \mathbb{S}_N$. Hence, $C_N$ is the eigenspace of $\mathbb{S}_N$ corresponding to the eigenvalue $+1$. Its dimension depends on the number of independent generators $S_i$ of the stabilizer or, equivalently, the number of elements of $\mathbb{S}_N$: if $\mathbb{S}_N$ has $2^{N-k}$ elements for some $0\leq k < N$, then $\dim C_N=2^k$. A stabilizer subspace of dimension $2^k$ might be used to encode $k$ logical qubits; the corresponding vectors belonging to $C_N$ are called quantum code words (for more details see Refs. \cite{GottesmanThesis,Steane96,Shor95}). 

In the particular case of $|\mathbb{S}_N|=2^N$, the subgroup stabilizes a unique state, known as the stabilizer state. Any stabilizer state is equivalent to a certain graph state under local unitary operations (see, e.g., Ref. \cite{Schlingemann}), and self-testing methods for graph states are already known \cite{Flavio,McKague2014}.
Our aim here is to go beyond the $k = 1$ case and provide device-independent certification methods for higher-dimensional subspaces $C_N$.
In order to exploit nonlocality as the resource for certification, a natural starting point are those stabilizers that generate genuinely entangled subspaces. In Supplemental Material A \cite{supp} we also provide 
a simple sufficient criterion to ascertain that a given stabilizer gives rise to a genuinely entangled subspace.

\emph{Self-testing of entangled subspaces.} 
Let us start off by providing the definition of self-testing of an entangled subspace. To this aim, let $\mathcal{H}_P$ be, as before, the prover's Hilbert space. Let then 
$\mathcal{H}_{P'}$ be a Hilbert space with dimension equal to that of the entangled subspace we want to 
self-test whereas $\mathcal{H}_{P''}$ some auxiliary Hilbert space such that 
$\dim \mathcal{H}_P= \dim\mathcal{H}_{P'}\dim \mathcal{H}_{P''}$. Notice that 
for the examples considered below $\mathcal{H}_{P'}=(\mathbbm{C}^2)^{\otimes N}$.

Let then $\ket{\phi}_{PE}$ be a purification of the mixed state $\rho_P$ shared by the provers to a larger Hilbert space $\mathcal{H}_{PE}=\mathcal{H}_P\otimes\mathcal{H}_E$, where  $\mathcal{H}_E$ represents all potential degrees of freedom which the provers do not have access to.

\begin{defin}\label{definition}
The behaviour $\mathcal{P}$ self-tests the entangled subspace spanned by the set of entangled states $\{\ket{\psi_i}\}_{i=1}^k$ if for any pure state $\ket{\phi}_{PE}\in\mathcal{H}_{PE}$ compatible with $\mathcal{P}$ through \eqref{correlations} one can deduce that: i) every local Hilbert space has the form $\mathcal{H}_{P_i}=\mathcal{H}_{P'_i}\otimes \mathcal{H}_{P_i''}$;
ii) there exists a local unitary transformation $U_P = U_1\otimes\cdots\otimes  U_N$ acting on $ \mathcal{H}_P $ such that

\begin{align}\label{stsub}
   U_P\otimes \mathbbm{1}_E(\ket{\phi}_{PE}) = \sum_{i} {c_i}\ket{\psi_i}_{P'}\otimes\ket{\xi_i}_{P''E}
\end{align}
for some normalised states $\ket{\xi_{i}}\in \mathcal{H}_{P''}\otimes \mathcal{H}_E$ and some nonnegative numbers $c_i\geq 0$ such that $\sum_i c_i^2 = 1$.
\end{defin}
Note that, if the set $\{\ket{\psi_i}\}_{i=1}^k$ spans a genuinely entangled subspace, then the state on the r.h.s. of Eq. (3) is genuinely multipartite entangled with respect to the partition  $\mathcal{H}_{P_1}|\mathcal{H}_{P_2}|\cdots|\mathcal{H}_{P_N}$ \footnote{Too see that, notice that upon performing a partial trace over the spaces $P''$ and $E$ and after a diagonalisation, the resulting state $\rho_{P'} = \protect\sum_i | \psi^{'}_{i} \protect\rangle \protect\langle \psi^{'}_{i} | $ is genuninely multipartite entangled, since it is fully supported on the subspace spanned by  $\{\protect\ket{\psi_i}\}_{i=1}^k$. It follows that the state $\rho_{P' P''}$ must also be genuinely multipartite entangled, since no GME state can be obtained after a partial trace of a non-GME one}. This implies that the above self-testing statement can also be used as a certificate of genuine multipartite entanglement for the measured state. Note also that, analogously to the case of state self-testing, a subspace can be certified from the observed correlations $\mathcal{P}$ up to certain equivalences such as local unitary transformations or extra degrees of freedom described by $\mathcal{H}_{P''}$ and $\mathcal{H}_E$. Interestingly, here we identify an additional degree of freedom encoded in the scalars $c_i$. The freedom of varying the values of those scalars implies that self-testing of an entangled subspace can also be understood as self-testing of all mixed states supported on that subspace and giving rise to the correlations $\mathcal{P}$.

In what follows, we show how to prove a self-testing statement according to the above definition based solely on the fact that the observed behaviour maximally violates a certain multipartite Bell inequality. As target subspaces we choose those used in quantum error correction. As quantum code words are highly entangled states, it is natural to expect them to display nonlocal correlations \cite{DVPeres,Walker}. Notice that for our purposes it is not enough to simply observe nonlocal correlations, but it is crucial to prove that states belonging to the subspaces of interest maximally violate a Bell inequality and such an inequality has to be carefully tailored to the considered code space. A couple of methods based on the stabilizer formalism that does the job was recently put forward in Refs. \cite{Srikanth2018,waegell2019benchmarks}. Here we provide an alternative construction,
inspired by Ref. \cite{Flavio}, that allows us to make a straightforward connection to the self-testing proof. {Before moving to the presentation of the results, it is important to emphasise that, as many prior works on self-testing, we assume that the source produces identical and independently distributed copies of the state $\rho_P$.} 

\emph{The five-qubit code.} The five-qubit code is the smallest possible code that corrects single-qubit errors \cite{5qubitA,5qubitB} on a logical qubit. It is also a stabilizer code, generated by the following four operators
acting on $(\mathbbm{C}^2)^{\otimes 5}$:
\begin{eqnarray}\label{5qubitStab}
\hspace{-0.5cm}&&S_1  = \sx^{(1)}\sz^{(2)} \sz^{(3)} \sx^{(4)},\qquad 
S_2  = \sx^{(2)} \sz^{(3)} \sz^{(4)} \sx^{(5)},  \nonumber\\
\hspace{-0.5cm}&&S_3 = \sx^{(1)} \sx^{(3)} \sz^{(4)} \sz^{(5)},\qquad \label{eq:S3}
S_4  = \sz^{(1)} \sx^{(2)}\sx^{(4)} \sz^{(5)} \label{eq:S4},
\end{eqnarray}
where $\sx^{(i)},\sz^{(i)}$ are the Pauli matrices acting on qubit $i$. One can check that the four operators above are independent and hence the code space, denoted $C_5$, is two dimensional, and, importantly, it is genuinely entangled (see Supplemental Material A \cite{supp} for a proof). 

In order to prove a self-testing statement for this subspace, we introduce a Bell inequality that is maximally violated by any pure state from $C_5$. To do so we build the inequality directly from the generators \eqref{5qubitStab}. For the first party we assign $\sx^{(1)} \rightarrow (\rA_0^{(1)} + \rA_1^{(1)})/\sqrt{2}$ and $\sz^{(1)} \rightarrow (\rA_0^{(1)} - \rA_1^{(1)})/\sqrt{2}$,
while for the remaining parties we simply replace $\sx^{(i)}\rightarrow \rA_0^{(i)}$ and $\sz^{(i)} \rightarrow \rA_1^{(i)}$, where $\rA^{(i)}_j$ are arbitrary binary observables (of unspecified but finite dimension) that are to be measured in a Bell experiment. Let then $\tilde{S}_i$ denote operators obtained from \eqref{5qubitStab} by making the above substitutions. 

Let us also define the following Bell inequality
\begin{eqnarray}\label{eq:I5}
I_5 &=&  \langle (\rA_0^{(1)} + \rA_1^{(1)}) \rA_1^{(2)} \rA_1^{(3)} \rA_0^{(4)} \rangle 
+\langle \rA_0^{(2)}\rA_1^{(3)}\rA_1^{(4)}\rA_0^{(5)} \rangle  \nonumber\\
&&+ \langle (\rA_0^{(1)} + \rA_1^{(1)})\rA_0^{(3)} \rA_1^{(4)} \rA_1^{(5)} \rangle  \nonumber\\
&&+ 2 \langle (\rA_0^{(1)} - \rA_1^{(1)})\rA_0^{(2)} \rA_0^{(4)} \rA_1^{(5)} \rangle \leq 5
\end{eqnarray}
whose local bound was directly computed by optimizing $I_5$ over deterministic strategies for which $\rA_{x_i}^{(i)} = \pm 1$. The above inequality is obtained by choosing a suitable linear combination of the expectation values of $\tilde{S}_i$, indeed it can be checked that $I_5 = \sqrt{2} ( \langle \tilde{S}_1 \rangle +  \langle \tilde{S}_3 \rangle ) +  \langle \tilde{S}_2 \rangle + 2 \sqrt{2}  \langle \tilde{S}_4 \rangle $.

The maximal quantum value of $I_5$ can also be straightforwardly determined and it amounts to $\beta_q=4\sqrt{2}+1$. To see that such value can be achieved by any state from $C_5$, let us notice that by making the following measurement choices
\begin{align}\label{eq:measurements}
\rA_0^{(1)}  = \frac{\sx+\sz}{\sqrt{2}} \, , &\qquad    \rA_1^{(1)}   = \frac{\sx-\sz}{\sqrt{2}},
\end{align}
for the first party and $\rA_0^{(i)}  = \sx$ and $\rA_1^{(i)}  = \sz$  $(i = 2,\ldots,5)$ for the remaining ones, 
$I_5$ becomes the expectation value of the Bell operator:
$\mathcal{B}'_5  = \sqrt{2}({S}_1 + {S}_2  + 2{S}_4)+ {S}_3 $.
It follows that its maximal eigenvalue is $4\sqrt{2}+1$ and it is precisely associated with the eigenspace stabilized by the four generators $S_i$ given in Eq. (\ref{5qubitStab}). 

To prove that there does not exist a quantum state and observables giving rise to a higher violation of $I_5$ it is enough 
to show that the following decomposition 
\begin{align} \label{eq:SOS5}
\beta_q \openone - \mathcal{B}_5  = & \frac{1}{\sqrt{2}} \left(\openone - \tilde{S}_1 \right)^2 + \frac{1}{2} \left(\openone -  \tilde{S}_2 \right)^2 \nonumber \\
& + \frac{1}{\sqrt{2}} \left(\openone -  \tilde{S}_3 \right)^2 + \sqrt{2} \left( \openone -  \tilde{S}_4 \right)^2,
\end{align}
holds true, which implies that the eigenvalues of $\mathcal{B}_5$ do not exceed $\beta_q$ for any choice of local observables $\rA_{x_i}^{(i)}$.

The maximal violation of inequality (\ref{eq:I5}) allows one to make the following self-testing statement. 
\begin{fakt}\label{5qubit}
Any behaviour achieving the maximal quantum violation of $I_5$ self-tests the entangled subspace $C_5$.
\end{fakt}
\begin{proof}
Here we provide a sketch of the proof for illustrative purposes, deferring the details to Supplemental Material B \cite{supp}. Consider
a state $\ket{\phi}_{PE}\in\mathcal{H}_{PE}$ and observables $\rA_{x_i}^{(i)}$ acting on $\mathcal{H}_{P_i}$ that maximally violate inequality (\ref{eq:I5}). From the decomposition \eqref{eq:SOS5} one deduces that
\begin{equation}\label{Stab}
(\tilde{S}_i\otimes\mathbbm{1}_E)\ket{\phi}_{PE}=\ket{\phi}_{PE}    
\end{equation}
for $i=1,\ldots,4$, which can be used to prove the existence
of local unitary operations $U_i$ acting on $\mathcal{H}_{P_i}$ such that $U\tilde{S}_iU^\dagger = S_i\tp\mathbbm{1}_{P''}$ for $i = 1,\ldots,4$, where $U=U_1\otimes\ldots\otimes U_5$ and the operators $S_i$ are given in Eq. (\ref{5qubitStab}).
This allows us to rewrite Eq.~(\ref{Stab}) as $({S}_i\tp\mathbbm{1}_{P''E})\ket{\psi}_{PE} = \ket{\psi}_{PE}$ with $\ket{\psi}_{PE}=(U\otimes \mathbbm{1}_E)\ket{\phi}_{PE}$.
As we show in Supplemental Material B \cite{supp}, the most general state satisfying all these four conditions is exactly $c \ket{\psi_1}\tp\ket{\xi_{1}} + \sqrt{1-c^2 } \ket{\psi_2}\tp\ket{\xi_{2}}$, where $0\leq c \leq 1$,  $\ket{\psi_1}$ and $\ket{\psi_2}$ are two orthogonal five-qubit states spanning $C_5$ whereas $\ket{\xi_1}$ and $\ket{\xi_2}$ are some auxiliary quantum states from $\mathcal{H}_{P''}\otimes\mathcal{H}_E$. 
\end{proof}

\emph{The toric code.}---The toric code is a paradigmatic example of topological quantum error correction codes \cite{Kitaev1997}. It allows one to store two logical qubits in four multi-qubit pure states of arbitrarily large number of particles. The logical qubits can be associated with ground states of a $1/2$-spin model on a torus, that is, a two-dimensional spin square lattice with periodic boundary conditions in which qubits are associated with the edges  
(see Fig. \ref{fig:toric}).

The toric code is also a stabilizer code with two types of stabilizing operators: the vertex and plaquette operators
\begin{equation}\label{eq:sV}
S_v = \prod_{i \in v} \sx^{(i)} \,, \qquad  S_p = \prod_{i \in p} \sz^{(i)} \, .
\end{equation}
For each of the generators $S_v$ and $S_p$, the product runs over operators acting on qubits sharing the same vertex $v$ or plaquette $p$, respectively (see Fig. \ref{fig:toric}). 
The above generators are not all independent, since they satisfy $\prod_v S_v=\prod_p S_p=\mathbbm{1}$. By simple counting argument, it follows that the set of states stabilized by these operators spans a four-dimensional subspace, denoted $C^{\tor}_N$, for any choice of the lattice size $L$.

\begin{figure}
\centering
\includegraphics[width=\columnwidth]{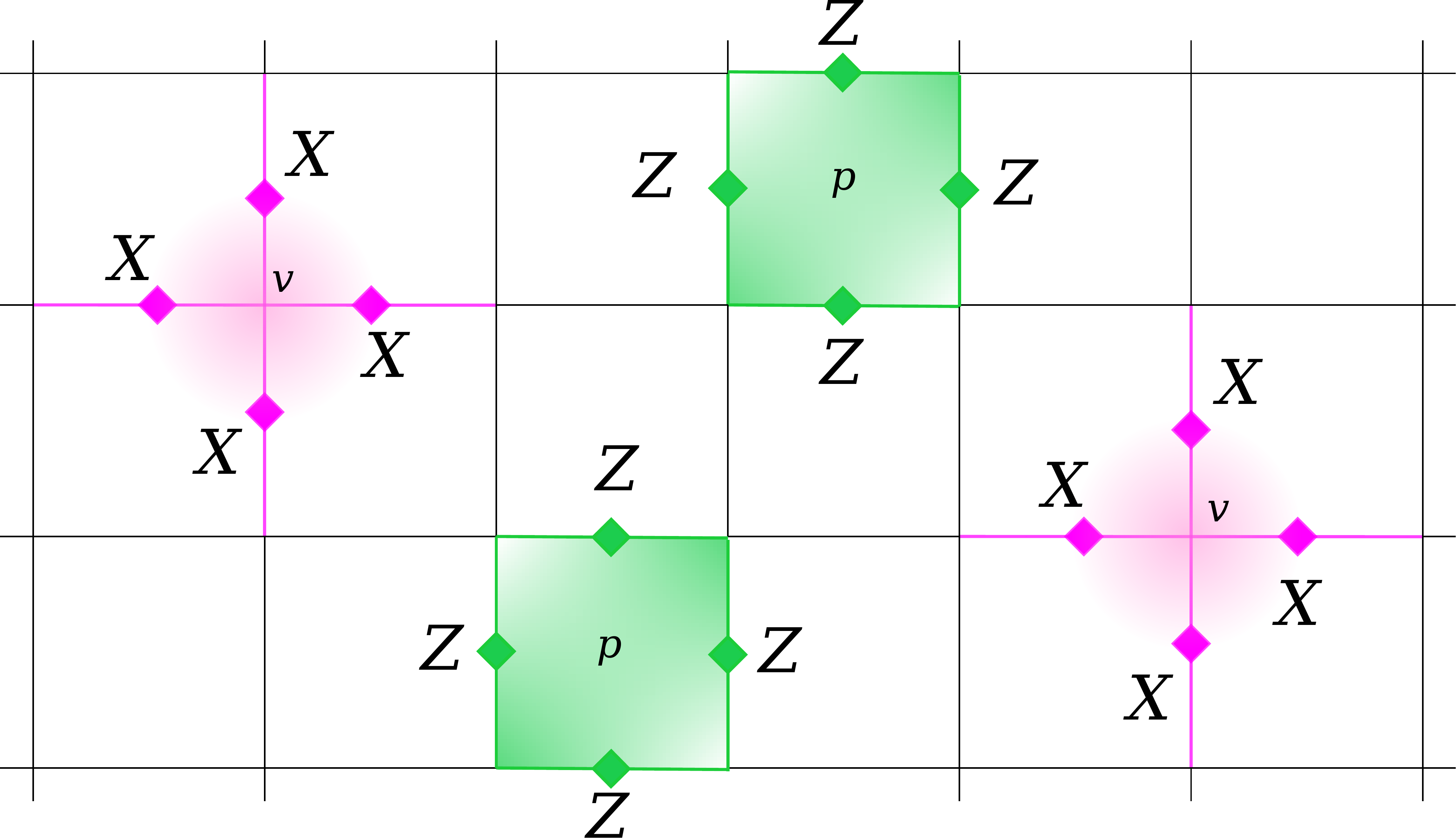}
\caption{\textbf{The toric code}. Each edge of the lattice represents a qubit. Stabilizing operators can be divided in two groups: those associated with each lattice vertex $v$ with $\sx$ acting on every qubit associated with an edge attached to the given vertex, or associated with each plaquette $p$ of the lattice with $\sz$ acting on each qubit represented by an edge surrounding the plaquette.} 
\label{fig:toric}
\end{figure}

The Bell inequality maximally violated by any mixed state supported in $C^{\tor}_N$ can be derived in a manner analogous to the previous example. For an arbitrarily chosen edge $j$, we substitute the Pauli operators $\sx^{(j)},\sz^{(j)}$ acting on the corresponding qubit with the combinations $(\rA^{(j)}_0 \pm \rA^{(j)}_1)/\sqrt{2}$, while for the other qubits we simply have $\sx^{(i)},\sz^{(i)} \rightarrow \rA_0^{(i)},\rA_1^{(i)}$ $(i\neq j)$. By applying this substitution to $S_v$ and $S_p$ we obtain operators $\tilde{S}_v$ and $\tilde{S}_p$ from which we obtain the following Bell inequality
\begin{align}\label{eq:belltoricmain}
I^{\tor}_N: =\sum_v \langle \tilde{S}_v \rangle + \sum_p \langle \tilde{S}_p \rangle\leq \beta_c^{\mathrm{tor}}(N).
\end{align}
It is not difficult to realize that its classical bound $\beta_c^{\mathrm{tor}}(N)$ amounts to $\beta_c^{\mathrm{tor}}(N) = 2\sqrt{2} + |p|-2 + |v|-2 = N-2\sqrt{2}(\sqrt{2}-1)$.
Moreover, as proven in the Supplemental Material B~\cite{supp}, the quantum bound is $\beta_q^{\mathrm{tor}}(N) = 4 + |p| + |v| -4 =  N>\beta_c^{\mathrm{tor}}(N)$. It follows that any pure state from $C_N^{\mathrm{tor}}$ achieves it, meaning that $C_N^{\mathrm{tor}}$ is an entangled subspace; in fact, in Supplemental Material A~\cite{supp} we prove it to be genuinely entangled. More importantly, 
as we prove in Supplemental Material B~\cite{supp}, the following self-testing statement can be made for it. 
\begin{fakt}\label{toric}
Any $N$-partite behaviour $\mathcal{P}$ achieving the maximal quantum violation of $I^{\tor}_N$ self-tests the entangled subspace $C^{\tor}_N$.
\end{fakt}

\emph{Noise robustness.} In practice, finite sampling effects and experimental errors render the maximal violation of a Bell inequality impossible to reach. For this reason, it is important to be able to make certification statements when the violation is not maximal. In Supplemental Material C~\cite{supp} we provide a framework allowing us to tackle this problem and use it to obtain a numerical indication of how robust our self-testing statement is for the five-qubit code. In particular, Fig.  \ref{fig:fidelity} shows how the numerically estimated subspace extractability, which we define to
quantify how close is the self-tested subspace to the desired one,
scales as a function of the observed Bell violation.

\begin{figure}
\includegraphics[width=0.45\textwidth]{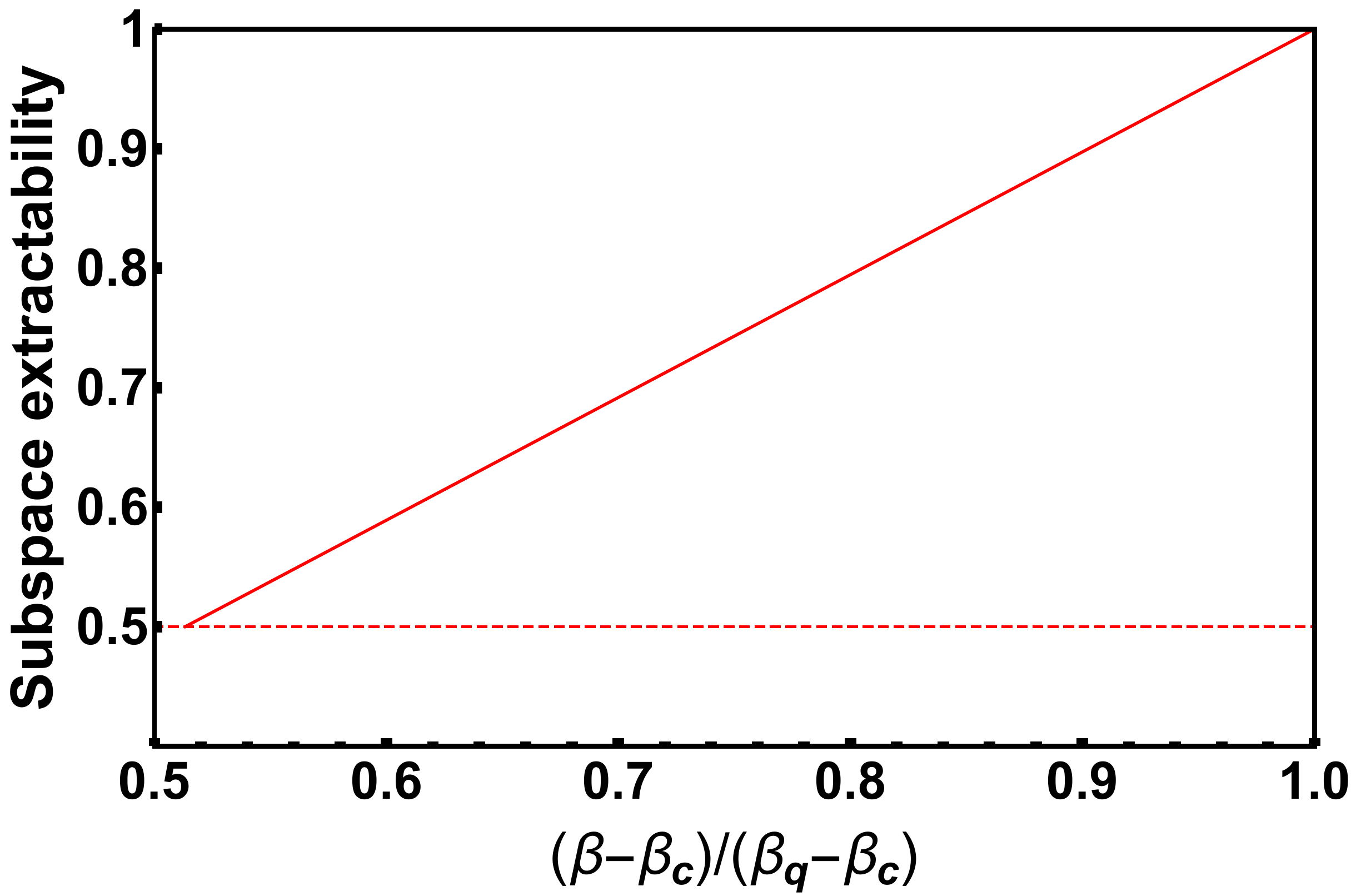}
\caption{Numerical estimation of the minimal subspace extractability for the target 5-qubit Shor's code subspace as a function of the relative observed violation $(\beta -  \beta_c)/(\beta_q - \beta_c)$ of the corresponding Bell inequality $I_5$.}
\label{fig:fidelity}
\end{figure}

\emph{Geometrical considerations.} 
A Bell inequality is generally believed to be useful for state self-testing only if its maximal violation can be associated with a single quantum state and to a single point in the set of quantum correlations \cite{Goh2018}.
Remarkably, subspace self-testing allows one to make non-trivial statements for Bell inequalities maximally violated by more than one correlation point.
As we show in in Supplemental Material D~\cite{supp}, the subspace of quantum correlations maximally violating inequalities \eqref{eq:I5} and \eqref{eq:belltoricmain} is spanned by two and three linearly independent correlation vectors, respectively.
Our results on subspace self-testing can thus be seen as complementary to the weaker form of self-testing presented in \cite{kaniewski2019weak}. While in \cite{kaniewski2019weak} the multiple correlations points associated to the self-testing statement are achieved by varying the measurement operators on the same state, we instead obtain a similar phenomenon by performing the same measurements on different quantum states.

\textit{Discussion.} 
We introduce the notion of self-testing of entangled subspaces---a device-independent method of certification that an entangled state belongs to a certain subspace of dimension at least two. By exploiting the formalism of stabilizer error correction codes we present two examples of multipartite Bell inequalities whose maximal quantum violation serves the purpose, that is, enables self-testing of entangled subspaces.
We also provide a framework to study the robustness of
subspace self-testing when the experimental imperfections are taken into account
and use it to investigate how robust is the self-testing statement for the five-qubit code. On a more fundamental level, our Bell inequalities can be used
to identify flat structures in the boundary of the sets of quantum correlations.

Our work opens a plethora of possibilities for future research.
By proposing a clear strategy to derive a Bell inequality and by being based on the broadly used stabilizer formalism, we believe that our self-testing technique has the potential to be generalized to other genuinely entangled stabilizer subspaces. From a broader point of view, it would be very interesting to understand which conditions a stabilizer subspace has to satisfy in order to be suitable for self-testing  (see \cite{Owidiusz} for a recent progress concerning this point).
We also believe that the presented techniques are a promising tool to quantify the average fidelity of quantum codeword preparations, thus opening the way to the device-independent certification of quantum error correction codes. Characterizing in more detail the robustness properties of our self-testing methods would also be essential for that.

\emph{Acknowledgements.}---We acknowledge discussions with O. Makuta.
F.B. acknowledges the support of the Spanish MINECO in the form of a Severo Ochoa PhD fellowship during his time at ICFO and the Deutsche Forschungsgemeinschaft (DFG, German Research Foundation) - Project number 414325145 in the framework of the Austrian Science Fund (FWF): SFB F71 during his time at MPQ. I.\v{S}. acknowledges support from SNSF (Starting grant DIAQ). 
R. A. acknowledges the support from the Foundation for Polish Science through the First Team project (First TEAM/2017-
4/31) co-financed by the European Union under the European Regional Development Fund.
A. A. acknowledges support from the ERC AdG CERQUTE, the AXA Chair in Quantum Information Science, the Government of Spain (FIS2020-TRANQI and Severo Ochoa CEX2019-000910-S), Fundació Cellex, Fundació Mir-Puig, Generalitat de Catalunya (CERCA, AGAUR SGR 1381 and QuantumCAT).

\bibliography{Subspaces}

\end{document}


\title{Device-independent certification of genuinely entangled subspaces: Supplemental Material}
\author{Flavio Baccari}
\email{flavio.baccari@mpq.mpg.de}
\affiliation{Max-Planck-Institut f\"ur Quantenoptik, Hans-Kopfermann-Stra{\ss}e 1, 85748 Garching, Germany}
\author{Remigiusz Augusiak}
\email{augusiak@cft.edu.pl}
\affiliation{Center for Theoretical Physics, Polish Academy of Sciences, Aleja Lotnik\'{o}w 32/46, 02-668 Warsaw, Poland}
\author{Ivan \v{S}upi\'{c}}
    \affiliation{D\'{e}partement de Physique Appliqu\'{e}e, Universit\'{e} de Gen\`{e}ve, 1211 Gen\`{e}ve, Switzerland}
\author{Antonio Ac\'{i}n}
\affiliation{ICFO-Institut de Ciencies Fotoniques, The Barcelona Institute of Science and Technology, 08860 Castelldefels (Barcelona), Spain}
\affiliation{ICREA - Instituci\'{o} Catalana de Recerca i Estudis Avancats, 08011 Barcelona, Spain}

\maketitle

In this Supplemental Material we prove that the subspaces considered in the main text contain only genuinely entangled states. We also provide rigorous proofs for the self-testing statements of Fact 1 and 2 in the main text. Furthermore, we analyse the robustness to noise of our self-testing statements and lastly, we study the geometrical structure of the faces maximally violating the inequalities $I_5$ and $I_N^{tor}$ presented in the main text.

\section{APPENDIX A: PROVING THAT SUBSPACES $C_5$ AND $C_N^{\mathrm{tor}}$ ARE GENUINELY ENTANGLED}

Here we prove that the subspaces corresponding to the five-qubit
code $C_5$ and the toric code $C_N^{\mathrm{tor}}$ are genuinely entangled, 
that is, contain only multipartite genuinely entangled states. 

We begin by providing a simple sufficient criterion for 
a subspace generated by a stabilizer $\mathbb{S}$ 
to be genuinely entangled. To this end, assume that $\mathbb{S}$ is generated by a set of $k$ stabilizing operators $\{S_i\}_{i=1}^k$. 
Consider then a bipartition of $N$ parties into 
two disjoint and nonempty groups $G$ and $G'$ such that $|G|+|G'|=N$, and denote by $\mathcal{G}$ the set of all such bipartitions. Given a bipartition $G|G'$, every stabilizing operator $S_i$ can be written as
\begin{equation}\label{division}
    S_i=S_i^G\otimes S_i^{G'}
\end{equation}
with $i=1,\ldots,k$, where each $S_i^{G}$ $(S_i^{G'})$ acts on the Hilbert space 
associated to the group $G$ ($G'$). Notice that due to the fact that for every pair $i\neq j$, the operators $S_i$ and $S_j$ commute, either 
\begin{equation}
 [S_i^G,S_j^G]=0\quad\mathrm{and}\quad   [S_i^{G'},S_j^{G'}]=0,
\end{equation}
or
\begin{equation}\label{antiCom}
 \{S_i^G,S_j^G\}=0\quad\mathrm{and}\quad   \{S_i^{G'},S_j^{G'}\}=0.
\end{equation}

Let us now formulate our criterion.

\begin{fakt}\label{fakt3}
Consider a stabilizer $\mathbb{S}$ generated by a set of stabilizing operators $S_i$ $(i=1,\ldots,k)$. If for every bipartition of $N$ parties into two disjoint and nonempty subsets $G$ and $G'$ 
there exist $i\neq j$ such that $\{S_i^{G},S_j^{G}\}=0$, then 
the subspace $C_N$ stabilized by $S_i$ is genuinely multipartite entangled.
\end{fakt}
\begin{proof} Let us begin by assuming that the subspace $C_N$ is not genuinely entangled, which means that it contains a pure state $\ket{\psi}$ such that 
\begin{equation}
    \ket{\psi}=\ket{\psi_G}\otimes\ket{\psi_{G'}}
\end{equation}
for some bipartition $G|G'$. From the fact that $S_i\ket{\psi}=\ket{\psi}$
we infer that with respect to this bipartition
\begin{equation}
    S_i^G\ket{\psi_G}=\mathrm{e}^{\mathrm{i}\varphi_i}\ket{\psi_G}
\end{equation}
for some $\varphi_i\in\mathbbm{R}$ (analogous identities hold true for the $G'$ group). This contradicts the fact that there exist $i\neq j$ such that 
\begin{equation}
    \{S_i^G,S_j^G\}\ket{\psi_G}=0,
\end{equation}
which completes the proof.
\end{proof}
Notice that, as $S_i$ mutually commute, the anticommutation relation in Fact \ref{fakt3} might also be formulated for $G'$.

\subsection{Subspace $C_5$}

Let us now illustrate the power of our criterion by applying it to the subspace $C_5$ corresponding to the five-qubit code. Recall that the corresponding stabilizing
operators are 
\begin{eqnarray}\label{5qubitStabApp}
\hspace{-0.5cm}&&S_1  = \sx^{(1)}\sz^{(2)} \sz^{(3)} \sx^{(4)},\qquad 
S_2  = \sx^{(2)} \sz^{(3)} \sz^{(4)} \sx^{(5)},  \nonumber\\
\hspace{-0.5cm}&&S_3 = \sx^{(1)} \sx^{(3)} \sz^{(4)} \sz^{(5)},\qquad 
S_4  = \sz^{(1)} \sx^{(2)}\sx^{(4)} \sz^{(5)}.
\end{eqnarray}
In the five-partite case, the relevant bipartitions can be divided into two possibilities: one party versus the rest (five such bipartitions) and two parties versus three (ten such cases). 

In the first case one notices that for every bipartition $G|G'$ with a single-element set $G=\{k\}$ there always exists a pair of the stabilizing operators such that $S_i^G=X^{(k)}$ and $S_j^{G}=Z^{(k)}$ and the condition (\ref{antiCom}) is satisfied. 

Let us then consider the second case. By direct check one realizes that for every bipartition $G|G'$ with $G=\{i,j\}$ such that $i<j$ there exists a stabilizing operator $S_m$ for which $S^G_m=\sx^{(i)}\sz^{(j)}$ $(i<j)$ and another stabilizing operator $S_n$ $(m\neq n)$ such that $S_n^G=\sx^{(j)}$ or
$S_n^G=\sz^{(i)}$ or $S_n^G=\sx^{(i)}\sx^{(j)}$, or finally $S_n^G=\sz^{(i)}\sz^{(j)}$.

\subsection{Subspace $C_{N}^{\mathrm{tor}}$} 

Here we show, employing Fact \ref{fakt3}, that the four-dimensional subspaces identified by the toric code consist of only genuinely entangled states. To this aim, it is enough to find, for every bipartition $G|G'$, two stabilizing operators that anticommute when restricted to $G$ or $G'$. Recall that the stabilizer associated to the toric code is generated by the following operators
\begin{equation}\label{app:toricgen}
S_v = \prod_{i \in v} \sx^{(i)} \,, \qquad  S_p = \prod_{i \in p} \sz^{(i)} \, .
\end{equation}
where one defines a different $S_v$ for every vertex and $S_p$ for every plaquette in the lattice. In the following we provide a constructive proof that such a set of stabilizing operators satisfies the assumption of Fact \ref{fakt3} for a lattice of any size. Before doing that, however, let us analyse some basic properties of this stabilizer. Notice that plaquette and vertex operators are composed of products of different Pauli matrices ($\sz$ and $\sx$ respectively) that, if taken independently, anticommute with each other. However, the products in \eqref{app:toricgen} are chosen carefully in order to define stabilizing operators that mutually commute. This happens because for every pair of plaquette and vertex operator, the subsets of qubits on which they act non-trivially are either disjoint or they overlap on exaclty two particles. In particular, the latter is exactly the case of a plaquette and a vertex generator acting non-trivially on the same pair of neighbouring qubits. For example, as shown in Figure \ref{fig:toricF}, 
if $j$ is a vertical edge qubit, $S_{v,\uparrow}^{(j)}$ and $S_{p,\rightarrow}^{(j)}$ would be two mutually commuting operators composed of a pair of anticommuting Pauli matrices acting on $j$ and its neighbour $j+L$. Notice that here we are adopting the notation introduced in Appendix B to denote stabilizing operators acting non-trivially on a given qubit. 

We are now ready to go back to proving that the toric code meets the assumptions of Fact \ref{fakt3}, i.e., given every bipartition $G|G'$, one can find two stabilizing operators that anticommute when restricted to the subset $G$. Thanks to the analysis of the commuting properties of the generators \eqref{app:toricgen} conducted above, it is now clear how to find two such operators. Namely, it suffices to find a pair of neighbouring qubits $i$ and $j$, belonging to a different subset of the bipartition, i.e., $i\in G$ and $j\in G'$ or \textit{vice versa}. Then, the two anticommuting operators will be the plaquette and vertex operator acting non-trivially on both the qubits restricted to, for example, the subset containing $i$. Indeed, since $j$ does not belong to the same subset, it follows that the non-trivial support of the two restricted generators now overlaps only on qubit $i$, where the two operators act with the Pauli matrix $\sx$ and $\sz$ respectively.

Therefore, proving that the subspace stabilized by the toric code is GME for any number of partices reduces to showing that such a pair of anticommuting operators can be found for every nontrivial bipartition $G|G'$ with $G,G'\neq \emptyset$. 

To see that this is indeed the case it is enough to realize that
for every bipartition $G|G'$ there exist a vertex operator $S_v$ which is divided between $G$ and $G'$ in the sense that at least one of the qubits $S_v$ acts on belongs to $G$ and at least one to $G'$. To prove this last statement, assume that such a vertex does not exist. Then, pick a vertex whose all qubits belong to, say $G$; recall that by assumption $G$ is not empty. Then all vertices connected to the chosen vertex by a qubit must also belong to $G$; recall that we assumed that there is no vertex divided between the two sets $G$ and $G'$.
Taking into account that all vertices in lattice are connected, following the above reasoning it is not difficult to realize that in fact all vertices must 
belong to $G$, meaning that $G'$ is empty. This, however, contradicts the
assumption that the bipartition $G|G'$ is nontrivial.

Consider then a vertex $v$ which is divided between $G$ and $G'$. 
We consider two possibilities: either one qubit associated to $v$ belongs 
to $G$ or two. The third case of three qubits belonging to $G$ is equivalent to the first one because we can always consider $G'$ instead of $G$. In the first case, we consider one of the two plaquette operators $S_p$ which act nontrivially on the qubit belonging to $G$. Then, it clearly follows that $S_v$ and $S_p$ anticommute when restricted to $G$. The second case is slightly more involved:
one has to idenfity a plaquette operator $S_p$ which acts nontrivially on only one of the $G$ qubits connected to the vertex $v$. In such a case, one can indeed show that $S_v^G$ and $S_p^G$ satisfy the anticommutation relation. 
Let us show that one can always find such a plaquette operator, by considering the two possible subcases of the two $G$ qubits connected to the vertex $v$ being: (i) both horizontal or vertical qubits, or (ii) a horizontal and a vertical qubit each.
For (i), we can label the two qubits without loss of generality as $j$ and $j+2L$. Then, by adopting the notation introduced above, a valid choice for the required plaquette operator is either $S_{p,\leftarrow}^{(j)}$ or $S_{p,\rightarrow}^{(j)}$ (see Fig. \ref{fig:toricF}). 
Similarly, in the subcase (ii), we can take the two qubits to be $j$ and $j+L$, so that the desired plaquette operator becomes $S_{p,\leftarrow}^{(j)}$.
This ends the proof.

\section{APPENDIX B: PROOFS OF SELF-TESTING STATEMENTS}
\label{AppA}

Here we provide full proofs of the self-testing
statements for the subspaces considered in the main text. 
For completeness we first recall a very useful fact, proven already in 
Refs. \cite{Popescu1992,Jed2}, that is used in our proofs.
\begin{lem}\label{lem:qubit}
\cite{Popescu1992,Jed2} Consider two hermitian operators $\tX$ and $\tZ$ acting on a Hilbert space $\mathcal{H}$ of dimension $D<\infty$ and satistfying the idempotency property $\tX^2 = \tZ^2 = \emph{\openone}$, as well as the anticommutation relation $\lbrace \tX , \tZ \rbrace = 0$. Then, $\mathcal{H}=\mathbbm{C}^2\otimes\mathbbm{C}^d$ for some $d$ such that $D=2d$, and there exists a local unitary operator $U$ for which
\begin{align}\label{iden}
U \tX U^\dagger = \sx \otimes \emph{\openone}_d \, , 
\qquad U \tZ U^\dagger = \sz \otimes \emph{\openone}_d \, ,  
\end{align}
where $\sx$ and $\sz$ are the $2\times 2$ Pauli matrices introduced before and $\emph{\openone}_d$ is the
identity matrix acting on the $d$-dimensional auxiliary Hilbert space.
\end{lem}

This lemma provides a way to characterise the measurement observables from their commutation properties. If the state to be self-tested can be tomographically retrieved by using two Pauli measurements, proving Lemma \ref{lem:qubit} basically reduces self-testing to quantum state tomography.

\begin{proof}
First of all, the fact that $\tX$ and $\tZ$ square to identity and are hermitian implies that their eigenvalues are $\pm 1$. This means that they are also unitary and allows one to rewrite the anticommutation relation $\{\tX,\tZ\}=0$ as 
\begin{equation}\label{iden2}
    \tX \tZ \tX = - \tZ,
\end{equation}
or as
\begin{equation}
    \tZ \tX \tZ = - \tX.
\end{equation}
As both $\tX$ and $\tZ$ are 
unitary and hermitian, these identities imply that 
the eigenspaces of both these operators
corresponding to the eigenvalues $\pm 1$ have equal dimensions. 
This has two consequences. The first one is that the Hilbert space
they both act on is $\mathcal{H}=\mathbb{C}^2\otimes \mathbb{C}^d$ with finite $d$ such that $D=2d$. The second one is that there
exists a unitary $U:\mathcal{H}\to\mathcal{H}$ which brings
the eigenvectors $\ket{e_i^{\pm}}$ of $\tZ$ 
to the product form
\begin{align}\label{ProdBasis}
U \ket{e_i^+} & = \ket{0}\ket{g_i} \, ,\\
U \ket{e_i^-} & = \ket{1}\ket{g_i} \, ,
\end{align}
with $\ket{g_i}$ being some orthogonal basis in $\mathbb{C}^d$,
which means that 
\begin{equation}
    U \tilde{Z}U^{\dagger}=\sz\otimes\mathbbm{1}_d.
\end{equation}
To obtain (\ref{iden}) for the $\tilde{X}$ operator 
it is enough to notice that Eq. (\ref{iden2}) implies that 
$\tilde{X}$ exchanges the eigenvectors of $\tilde{Z}$
corresponding to different eigenvalues, that is, 
\begin{equation*}
\ket{e_i^-} = \tX \ket{e_i^+}.    
\end{equation*}
Thus, in the product basis (\ref{ProdBasis}), $\tilde{X}$ is of the form
\begin{equation}
    U\tilde{X}U^{\dagger}=\sx\otimes \mathbb{1}_d, 
\end{equation}
which completes the proof.
\end{proof}

\subsection{Self-testing the five-qubit code subspace}

In this section we provide a detailed proof of self-testing of the $C_5$ subspace. 
To this aim, let us assume that a state $\ket{\phi}\in\mathcal{H}_{PE}$ and observables $\rA_{x_i}^{(i)}$ acting on $\mathcal{H}_{P_i}$   maximally violate our inequality $I_5$, which we recall here for completeness
\begin{eqnarray}\label{eq:I5}
I_5 &=&  \langle (\rA_0^{(1)} + \rA_1^{(1)}) \rA_1^{(2)} \rA_1^{(3)} \rA_0^{(4)} \rangle 
+\langle \rA_0^{(2)}\rA_1^{(3)}\rA_1^{(4)}\rA_0^{(5)} \rangle  \nonumber\\
&&+ \langle (\rA_0^{(1)} + \rA_1^{(1)})\rA_0^{(3)} \rA_1^{(4)} \rA_1^{(5)} \rangle  \nonumber\\
&&+ 2 \langle (\rA_0^{(1)} - \rA_1^{(1)})\rA_0^{(2)} \rA_0^{(4)} \rA_1^{(5)} \rangle \leq 5 \, \, .
\end{eqnarray}
Then, making the following substitutions
\begin{equation}
    \tilde{X}^{(1)} = \frac{1}{\sqrt{2}}\left[\rA_0^{(1)} + \rA_1^{(1)}\right],\quad
     \tilde{Z}^{(1)} = \frac{1}{\sqrt{2}}\left[\rA_0^{(1)} - \rA_1^{(1)}\right],
\end{equation}
and $\tilde{X}^{(i)}=\rA_0^{(i)}$ and $\tilde{Z}^{(i)}=\rA_1^{(i)}$
for $i=2,\ldots,5$, the operators $\tilde{S}_i$ introduced already 
in the main body of our work can be stated as
\begin{eqnarray}\label{5qubitStabII}
\hspace{-0.5cm}&&\tilde{S}_1  = \tilde{X}^{(1)}\tilde{Z}^{(2)} \tilde{Z}^{(3)} \tilde{X}^{(4)},\qquad 
\tilde{S}_2  = \tilde{X}^{(2)} \tilde{Z}^{(3)} \tilde{Z}^{(4)} \tilde{X}^{(5)},  \nonumber\\
\hspace{-0.5cm}&&\tilde{S}_3 = \tilde{X}^{(1)} \tilde{X}^{(3)} \tilde{Z}^{(4)} \tilde{Z}^{(5)},\qquad \label{eq:S3}
\tilde{S}_4  = \tilde{Z}^{(1)} \tilde{X}^{(2)}\tilde{X}^{(4)} \tilde{Z}^{(5)}. \nonumber\\ \label{eq:S4}
\end{eqnarray}
Let us also recall the sum-of-squares decomposition for $I_5$, namely
\begin{align} \label{eq:SOS5}
\beta_q \openone - \mathcal{B}_5  = & \frac{1}{\sqrt{2}} \left(\openone - \tilde{S}_1 \right)^2 + \frac{1}{2} \left(\openone -  \tilde{S}_2 \right)^2 \nonumber \\
& + \frac{1}{\sqrt{2}} \left(\openone -  \tilde{S}_3 \right)^2 + \sqrt{2} \left( \openone -  \tilde{S}_4 \right)^2,
\end{align}
From the decomposition \eqref{eq:SOS5} we deduce that the operators above satisfy the following conditions
\begin{equation}\label{conditions}
\tilde{S}_i\ket{\phi}=\ket{\phi}.
\end{equation}
where to simplify the notation, here and below we omit the identity acting on $\mathcal{H}_E$.
%
%

Let us now prove, using relations (\ref{conditions}), that for every $i$, the operators $\tilde{X}^{(i)}$ and $\tilde{Z}^{(i)}$ satisfy the conditions of Lemma \ref{lem:qubit} on the support of $\rho_i$ with the latter being the reduced density matrix of $\ket{\phi}$ corresponding to the Hilbert space $\mathcal{H}_{PE}$. First, it is direct to see that, by the very construction, $\tilde{X}^{(1)}$ and $\tilde{Z}^{(1)}$ satisfy the anticommutation relation
\begin{equation}\label{AntiCom1}
    \{\tilde{X}^{(1)},\tilde{Z}^{(1)}\}=0.
\end{equation}
To prove that they also square to identity on the support of $\rho_1$ we use the condition (\ref{conditions}) for $i=1$ and $i=4$, obtaining
\begin{equation}
    \tilde{S}_1^2\ket{\phi}=\tilde{S}_4^2\ket{\phi}=\ket{\phi},
\end{equation}
which due to the fact that $[\tilde{X}^{(i)}]^2=[\tilde{Z}^{(i)}]^2=\mathbbm{1}$ for $i=2,\ldots,4$ immediately imply that 
$[\tilde{X}^{(1)}]^2=[\tilde{Z}^{(1)}]^2=\mathbbm{1}$ on the support of $\rho_1$. 

Let us now prove the anticommutation relations for the remaining operators $\tilde{X}^{(i)}$ and $\tilde{Z}^{(i)}$  $(i=2,\ldots,5)$; notice that, by definition, 
they already satisfy $[\tilde{X}^{(i)}]^2=[\tilde{Z}^{(i)}]^2=\mathbbm{1}$.
To this end, we rewrite the conditions (\ref{conditions}) for $i =1,4$ as
\begin{align}
    \tilde{X}^{(2)} \ket{\phi} & = \tilde{Z}^{(1)} \tilde{X}^{(4)} \tilde{Z}^{(5)} \ket{\phi}, \nonumber \\
    \tilde{Z}^{(2)} \ket{\phi} & = \tilde{X}^{(1)} \tilde{Z}^{(3)} \tilde{X}^{(4)}  \ket{\phi},
\end{align}
which leads us to
\begin{equation}
    \lbrace \tilde{X}^{(2)} , \tilde{Z}^{(2)} \rbrace \ket{\phi} = \lbrace \tilde{X}^{(1)} , \tilde{Z}^{(1)} \rbrace \tilde{Z}^{(3)} \tilde{Z}^{(5)}  \ket{\phi} = 0,
\end{equation}
where the last equality is a consequence of Eq. (\ref{AntiCom1}). In a similar way one can exploit the operator relations stemming from Eq. (\ref{conditions})
to obtain anticommutation relations for the remaining three sites.
Indeed, we can combine the conditions (\ref{conditions}) for $\tilde{S}_3$ and $\tilde{S}_4$ to get
\begin{equation}
 \lbrace \tilde{X}^{(4)} , \tilde{Z}^{(4)} \rbrace \ket{\phi} = \lbrace \tilde{X}^{(1)} , \tilde{Z}^{(1)} \rbrace \tilde{X}^{(2)} \tilde{X}^{(3)}  \ket{\phi} = 0 \, ,
\end{equation}
and combine $\tilde{S}_2$ and $\tilde{S}_4$ to obtain
\begin{equation}
 \lbrace \tilde{X}^{(5)} , \tilde{Z}^{(5)} \rbrace \ket{\phi} = \lbrace \tilde{X}^{(4)} , \tilde{Z}^{(4)} \rbrace \tilde{Z}^{(1)} \tilde{Z}^{(3)}  \ket{\phi} = 0 \, .
\end{equation}
Finally, we combine the conditions arising from $\tilde{S}_2$ and $\tilde{S}_3$ to show
\begin{equation}
 \lbrace \tilde{X}^{(3)} , \tilde{Z}^{(3)} \rbrace \ket{\phi} = \lbrace \tilde{X}^{(5)} , \tilde{Z}^{(5)} \rbrace \tilde{Z}^{(1)} \tilde{Z}^{(2)}  \ket{\phi} = 0 \, .
\end{equation}

Given the above anticommutation relations for operators acting on all sites, we can now make use of Lemma \ref{lem:qubit} which allows us to introduce the local unitary operations 
$U_i:\mathcal{H}_{P_i}\mapsto \mathcal{H}_{P_i}$
that map the operators $\tZ^{(i)}, \tX^{(i)}$ to the qubit Pauli matrices in the sense that
\begin{equation}\label{eq:Paulirot}
U_i \tZ^{(i)} U_i^\dagger = \sz^{(i)} \otimes \openone_{P''_i}  \, , \quad  U_i \tX^{(i)} U_i^\dagger = \sx^{(i)} \otimes \openone_{P''_i} \, ,
\end{equation}
with $i = 1,\ldots,5$. If we define $\ket{\psi} = (U\otimes \mathbbm{1}_E) \ket{\phi}$, where $U = U_1 \otimes \ldots \otimes U_5$, it follows from the conditions \eqref{conditions} and the transformations \eqref{eq:Paulirot}, that
\begin{equation}\label{eq:stabqub}
S_i \otimes \openone_{P''E}\ket{\psi} = \ket{\psi} \qquad (i = 1,\ldots,4) \, ,    
\end{equation}
where $S_i$ are the stabilizer operators of $C_5$, defined in Eq.~\eqref{5qubitStabApp}.

Let us finally prove that the most general form of a state $\ket{\psi}\in(\mathbbm{C}^{2})^{\otimes 5}\otimes\mathcal{H}_{P''E}$ compatible with conditions (\ref{eq:stabqub}) is 
\begin{equation}
    \ket{\psi}=c\ket{\psi_1}\otimes \ket{\xi_1}+\sqrt{1-c^2}\ket{\psi_2}\otimes \ket{\xi_2},
\end{equation}
where $\ket{\psi_i}$ are two orthonormal states spanning $C_5$, 
$\ket{\xi_i}\in\mathcal{H}_{P''E}$ 
are some auxiliary states and $c\in[0,1]$.

To this aim, we consider
the Schmidt decomposition of $\ket{\psi}$ with respect to the tensor
product of the five-qubit Hilbert space $(\mathbbm{C}^{2})^{\otimes 5}$ and
$\mathcal{H}_{P''E}$, 
\begin{equation}\label{Schmidt}
    \ket{\psi}=\sum_{j}\lambda_j \ket{\eta_j}\ket{\varphi_j},
\end{equation}
where $\ket{\eta_j}\in (\mathbbm{C}^{2})^{\otimes 5}$ as well as $\ket{\varphi_j}\in\mathcal{H}_{P''E}$ are some orthogonal vectors 
in the corresponding Hilbert spaces. By plugging Eq. (\ref{Schmidt}) 
into Eq. (\ref{eq:stabqub}) we arrive at
\begin{equation}
   \sum_{j}\lambda_j (S_i\ket{\eta_j})\ket{\varphi_j} =
   \sum_{j}\lambda_j \ket{\eta_j}\ket{\varphi_j},
\end{equation}
which after projecting onto $\ket{\varphi_j}$ gives
$S_i\ket{\eta_j}=\ket{\eta_j}$ for every $i=1,2,3,4$. 
This means that each vector $\ket{\eta_j}$ belongs to the 
subspace $C_5$ and therefore it might be written as 
$\ket{\eta_j}=\alpha_j\ket{\psi_1}+\beta_j\ket{\psi_2}$
for some $\alpha_j$ and $\beta_j$ such that $|\alpha_j|^2+|\beta_j|^2=1$.
Substituting this last form into Eq. (\ref{Schmidt}), directly leads
to 
\begin{equation}
    \ket{\psi}=\ket{\psi_1}|\tilde{\xi}_1\rangle+\ket{\psi_2}|\tilde{\xi}_2\rangle
\end{equation}
with some in general unnormalized vectors $|\tilde{\xi}_j\rangle\in\mathcal{H}_{P''E}$.
Clearly, this can be rewritten as
\begin{equation}
    \ket{\psi}=c\ket{\psi_1}|\xi_1\rangle+\sqrt{1-c^2}\ket{\psi_2}
    |\xi_2\rangle,
\end{equation}
where $c\in[0,1]$ and $|\xi_i\rangle$ are now normalized.
This completes the proof.

\subsection{Self-testing the toric code subspace}

\begin{figure*}
\centering
\includegraphics[width=0.7\textwidth]{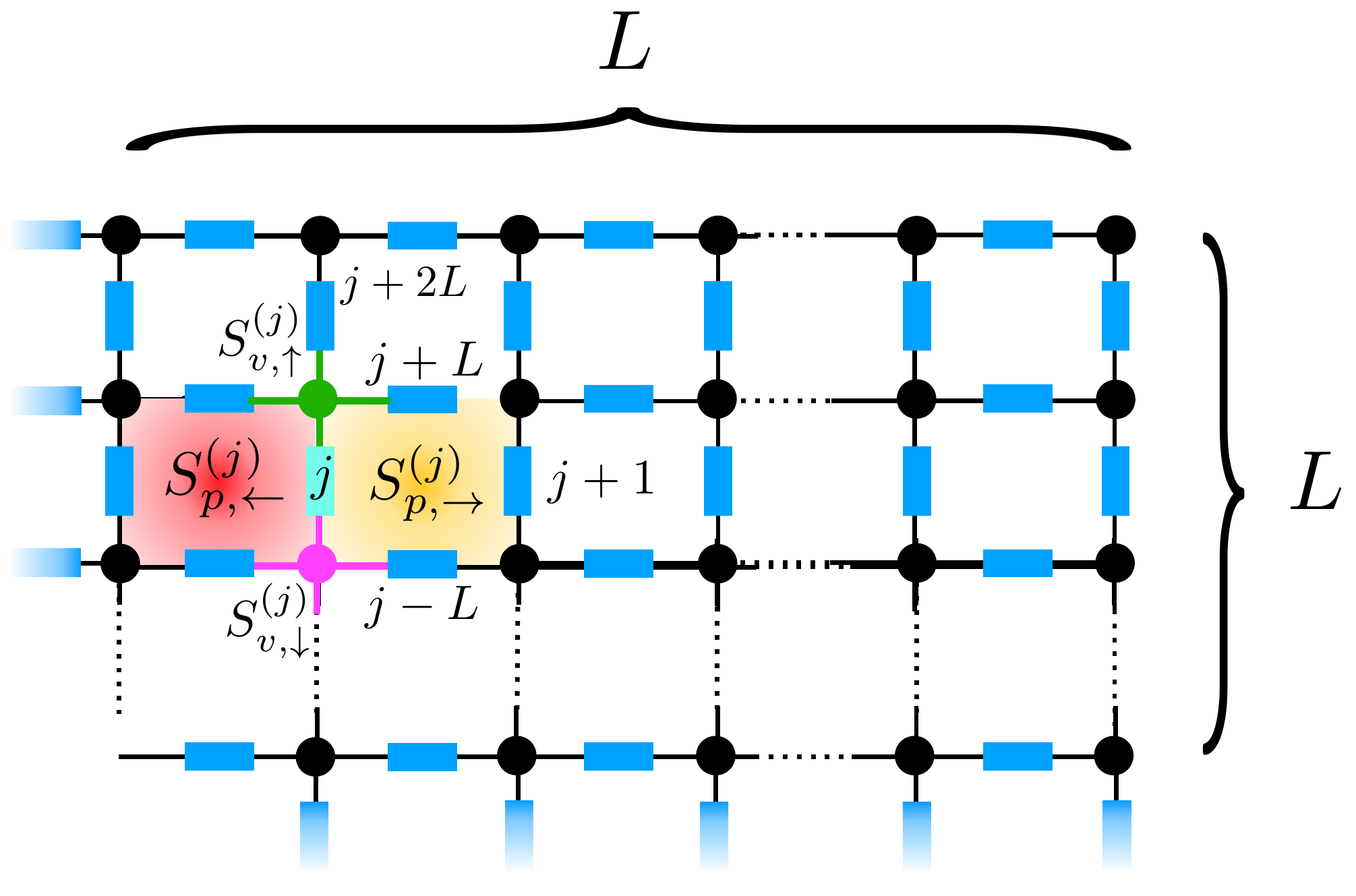}
\caption{Graphical representation of the toric code model, where the light blue rectangles placed on each edge of the lattice represent qubits. In short, we refer to a qubit associated to a vertical (horizontal) edge as a vertical(horizontal) qubit. To label qubits, we choose a reference vertical qubits $j$ and start counting forward by moving to the next vertical qubit on the right. By proceeding in this way, we label qubits from $j+1$ to $j+L-1$. After that, we move to the horizontal qubit on the top-right side of qubit $j$ and label it $j+L$, proceeding then forward with the next horizontal qubit to the right, until arriving at the qubit to the left of $j+L$. It then follows that the vertical qubit on top of $j$ would be labelled as $j+2L$.
Similarly, we can label plaquettes and vertex stabilizing operators depending on where are located with respect to a given qubit, as shown in the Figure for the case of qubit $j$. An analogous procedure can be done to label operators with respect to a horizontal qubit. For instance, we could refer to $S_{p,\rightarrow}^{(j)}$ as $S_{p,\downarrow}^{(j +L)}$, or to $S_{v,\uparrow}^{(j)}$ as $S_{v,\leftarrow}^{(j+L)}$.\label{fig:toricF}}
\end{figure*}

Let us begin by writing explicitly 
the Bell inequality maximally violated by every state belonging
to $C_N^{\mathrm{tor}}$. Recall to this end
that every qubit $i$ is contained in two plaquettes and is connected to two vertices, which we identify as $p^{(i)}_1,p^{(i)}_2$ and $v^{(i)}_1,v^{(i)}_2$ respectively. Now, our Bell inequality reads
\begin{equation}\label{eqapp:Itor}
    I_{N}^{\mathrm{tor}}:=\langle \mathcal{B}_N^{\mathrm{tor}}\rangle=\sum_{v}\langle \tilde{S}_v\rangle+\sum_{p}\langle \tilde{S}_p\rangle\leq \beta_c^{\mathrm{\tor}}(N),
\end{equation}
where the classical bound reads 
\begin{equation}
    \beta_c^{\mathrm{tor}}(N)=N-2\sqrt{2}(\sqrt{2}-1),
\end{equation}
whereas the Bell operator is given by
\begin{eqnarray}\label{eq:belltoric}
\mathcal{B}^{\tor}_N & = &\frac{1}{\sqrt{2}}\sum_{k = 1,2} (\rA_0^{(j)} + \rA_1^{(j)} ) \prod_{\substack{i \in v^{(i)}_k \\ i \neq j}} \rA_0^{(j)}  \nonumber \\  
&& + \frac{1}{\sqrt{2}}\sum_{k = 1,2} (\rA_0^{(j)} - \rA_1^{(j)} ) \prod_{\substack{i \in p^{(i)}_k \\ i \neq j}} \rA_1^{(j)} \nonumber \\
&& + \sum_{v \neq v^{(j)}_1,v^{(j)}_2} \prod_{i \in v} \rA_0^{(i)}
 + \sum_{p \neq p^{(j)}_1,p^{(j)}_2} \prod_{i \in p} \rA_1^{(i)},
\end{eqnarray}
where qubit $j$ is the chosen qubit for which our Bell expression contains 
combinations of observables.

The maximal quantum violation of this inequality amounts to 
\begin{eqnarray}\label{maxQtor}
    \beta_q^{\mathrm{tor}}(N) &=& 4  + |p| + |v| -4\nonumber\\
    &=& |p| + |v|=N
\end{eqnarray}
and can be achieved by the following measurements
\begin{equation}
\rA_0^{(j)} = \frac{\sx^{(j)}+\sz^{(j)}}{\sqrt{2}} \, , \qquad    \rA_1^{(j)}   = \frac{\sx^{(j)}-\sz^{(j)}}{\sqrt{2}},
\end{equation}
for the party $j$, and 
\begin{equation}
\rA_0^{(i)} =\sx^{(i)}  ,  \qquad   \rA_1^{(i)}  = \sz^{(i)} 
\end{equation}
for the remaining parties $i\neq j$, and every state belonging to $C^{\tor}_N$. 
It follows that \eqref{eq:belltoric} is violated for every $N$.

To prove that (\ref{maxQtor}) is indeed the maximal quantum value of $\mathcal{B}_N^{\mathrm{tor}}$ one checks that the following sum-of-squares decomposition holds true
\begin{widetext}
\begin{eqnarray}\label{eq:toricSOS}
\beta_q^{\mathrm{tor}}(N)\openone - \mathcal{B}^{\tor}_N & =& \frac{1}{2} \sum_{k = 1,2} \left( \openone -  \frac{\rA_0^{(j)} + \rA_1^{(j)}}{\sqrt{2}} \prod_{\substack{i \in v^{(j)}_k \\ i \neq j}} \rA_0^{(i)}  \right)^2  + \frac{1}{2} \sum_{k = 1,2} \left( \openone -  \frac{\rA_0^{(j)} - \rA_1^{(j)}}{\sqrt{2}} \prod_{\substack{i \in p^{(j)}_k \\ i \neq j}} \rA_1^{(i)}  \right)^2 \nonumber \\
&& + \frac{1}{2} \sum_{v \neq v^{(j)}_1,v^{(j)}_2} \left( \openone -  \prod_{i \in v} \rA_0^{(i)} \right)^2 + \frac{1}{2} \sum_{p \neq p^{(j)}_1,p^{(j)}_2} \left( \openone -  \prod_{i \in p} \rA_1^{(i)} \right)^2 \, .
\end{eqnarray}
\end{widetext}

Let us now move on to proving Fact 2 of the main text.
To this end, let us assume that an $N$-partite state $\ket{\phi}$ 
and observables $\rA_{x_i}^{(i)}$ $(x_i=0,1)$ maximally violate our Bell inequality, that is
\begin{equation}\label{condition}
    \langle\phi |\mathcal{B}_N^{\mathrm{tor}}|\phi\rangle = \beta_q^{\mathrm{tor}}(N)
\end{equation}
with $\mathcal{B}_N^{\mathrm{tor}}$ given in Eq. (\ref{eq:belltoric}).

Let us then introduce the operators
\begin{equation}\label{operators}
    \tilde{X}^{(j)} = \frac{\rA_0^{(j)} + \rA_1^{(j)}}{\sqrt{2}},
    \qquad \tilde{Z}^{(j)} = \frac{\rA_0^{(j)} - \rA_1^{(j)}}{\sqrt{2}}
\end{equation}
for the chosen party $j$, and $\tilde{X}^{(i)}, \tilde{Z}^{(i)} = \rA_0^{(i)}, \rA_1^{(i)}$ for all $i\neq j$. We now make use of the condition resulting from the sum-of-squares decomposition \eqref{eq:toricSOS} to prove that for every $i$, the operators $\tX^{(i)},\tZ^{(i)}$ anticommute and square to identity on the support of $\rho_i$ with the latter being the reduced density matrix of $\ket{\phi}$. Indeed, from Eqs. (\ref{eq:toricSOS}) and (\ref{condition}) one infers that
\begin{equation}\label{eq:stabVP}
    \tilde{S}_v \ket{\phi} = \tilde{S}_p \ket{\phi} = \ket{\phi} \, \quad \forall p,v .
\end{equation}

To do so, it is convenient to associate to each qubit in the lattice the subset of operators \eqref{app:toricgen} that act non-trivially on it, that is, those that contain $\sx$ or $\sz$ on this site. It is easy to see that for a two-dimensional lattice of any size there are precisely 
two vertex and two plaquette operators acting nontrivialy on each qubit. For later convenience, let us then introduce a notation for these four non-trivial generators: if the qubit $j$ corresponds to a horizontal edge, we identify the corresponding generators as
\begin{equation}\label{app:substab1}
\lbrace S_{v,\leftarrow}^{(j)},S_{v,\rightarrow}^{(j)},S_{p,\uparrow}^{(j)},S_{p,\downarrow}^{(j)} \rbrace \, ,
\end{equation}
where the arrows in the notation refer to where the vertex (plaquette) is located with respect to the qubit.
Similarly, if the qubit $j$ corresponds to a vertical edge, we denote the
corresponding four stabilizing operators as
\begin{equation}\label{app:substab2}
\lbrace S_{v,\uparrow}^{(j)},S_{v,\downarrow}^{(j)},S_{p,\leftarrow}^{(j)},S_{p,\rightarrow}^{(j)} \rbrace \, .
\end{equation}
Notice that, given two neighbouring qubits (i.e., qubits associated to edges connected to the same vertex), some of the elements in the two subsets in \eqref{app:substab1} and \eqref{app:substab2} are common (see Fig. \ref{fig:toricF} for an example).

Let us now focus on the qubit $j$ and consider a vertex and a plaquette operator that act non-trivially on $j$. Withouth loss of generality, we can take them to be $S_{v,\uparrow}^{(j)}$ and $S_{p,\rightarrow}^{(j)}$. From the stabilizing conditions (\ref{eq:stabVP}) for $\tilde{S}_{v,\uparrow}^{(j)}$ and  $\tilde{S}_{p,\rightarrow}^{(j)}$ as well as the fact that, by the very construction, the operators $\tilde{X}^{(i)}$ and $\tilde{Z}^{(i)}$ with $i\neq j$ square to identity, we obtain
\begin{equation}
\left(\tX^{(j)}\right)^2 \ket{\phi} = \left(\tZ^{(j)}\right)^2 \ket{\phi} =  \ket{\phi}.
\end{equation}
This directly implies that, as anticipated, both $\tX^{(j)}$
and $\tZ^{(j)}$ square to identity on the support of $\rho_j$. Then, by virtue
of Eq. (\ref{operators}), one directly sees that 
\begin{equation}
    \{\tilde{X}^{(j)},\tilde{Z}^{(j)}\}=0.
\end{equation}
Let us now move to the remaining pairs of operators 
$\tilde{X}^{(i)}$ and $\tilde{Z}^{(i)}$ with $i\neq j$ and prove
that they anticommute too; recall that by definition they square to identity.
We proceed in a recursive way. 

First, let us assume that $j$ is associated to a vertical edge and consider one of the neighbours of
$j$, denoted $j+L$ (cfr. Fig. \ref{fig:toricF}). To prove the anticommutation relation for the operators
acting on this qubit we assume, without any loss of generality that the plaquette and vertex operators shared between party $j$ and $j+L$ are exactly $\tilde{S}_{v,\uparrow}^{(j)}$ and  $\tilde{S}_{p,\rightarrow}^{(j)}$.
By making use of the corresponding stabilizing conditions \eqref{eq:stabVP} and the fact that all the local operators square to identity on the support of the state, we obtain the following equations
\begin{align}
\tX^{(j+L)} \ket{\phi} & = \tX^{(j)} \tX^{(j+2L)} \tX^{(j+2L-1)}  \ket{\phi}, \nonumber \\
\tZ^{(j+L)} \ket{\phi} & = \tZ^{(j)} \tZ^{(j+1)}  \tZ^{(j-L)}  \ket{\phi}, 
\end{align}
which lead us to
\begin{eqnarray}\label{anticommL}
\lbrace \tX^{(j+L)} , \tZ^{(j+L)} \rbrace \ket{\phi} &=& \lbrace \tX^{(j)} , \tZ^{(j)} \rbrace\tX^{(j+2L)} \tX^{(j+2L-1)} \nonumber\\
&&\times \tZ^{(j+1)}  \tZ^{(j-L)} \ket{\phi} \nonumber\\
&=& 0.
\end{eqnarray}

In exactly the same way we can prove the anticommutation relations for the remaining three neighbours of $j$ which share the plaquette and the vertex operators with it. To this end, it is enough to combine the stabilizing conditions (\ref{eq:stabVP}) corresponding to these operators for $j$ and its neighbour in the same way as in (\ref{anticommL}). 

One then realizes that the same approach can be applied to every neighbour of these four neighbours of $j$, and, thus recursively to every qubit of the lattice, which finally gives
\begin{equation}
    \{\tilde{X}^{(i)},\tilde{Z}^{(i)}\}=0
\end{equation}
for every $i$.

Having the anticommutation relations for all the parties, we make use of Lemma \ref{lem:qubit} which implies the existence of a unitary operations $U_i$ acting on $\mathcal{H}_{P_i}$ such that 
\begin{equation}
    U\tilde{X}^{(i)}U^{\dagger}=\sx^{(ji}\otimes \mathbbm{1}_{P''_i},\quad
    U\tilde{Z}^{(i)}U^{\dagger}=\sz^{(i)}\otimes \mathbbm{1}_{P''_i}
\end{equation}
for every $i$. Now, denoting $\ket{\psi} = U\otimes\mathbbm{1}_E \ket{\phi}$, we see that $\ket{\psi}$ satisfies the following stabilizing conditions 
\begin{align}
S_v \otimes \openone_{P''E} \ket{\psi} & = \ket{\psi} \, \quad \forall \, v \, , \nonumber \\
S_p \otimes \openone_{P''E} \ket{\psi} & = \ket{\psi} \, \quad \forall \, p \, ,
\end{align}
where $S_v$ and $S_p$ are the stabilizing operators given in \eqref{app:toricgen}.
Using the same method as in the case of the five-qubit code one can show that 
\begin{equation}
    \ket{\psi}=\sum_{i=1}^4c_i\ket{\psi_i}\otimes\ket{\xi_i},
\end{equation}
where $c_i$ are nonnegative numbers such that $c_1^2+c_2^2+c_3^2+c_4^2=1$, 
$\ket{\psi_i}$ are $N$-qubit vectors spanning $C_N^{\mathrm{tor}}$
and $\ket{\xi_i}\in\mathcal{H}_{P''E}$ are some auxiliary states. This completes the proof.

\section{Appendix C: Noise robustness bounds for subspace self-testing}

Here we address the problem of applying our self-testing protocols to realistic experimental scenarios. This requires being able to generalize self-testing statements to situations where, due to experimental imperfections, the maximal violation of the considered Bell inequalities cannot be attained.
In particular, suppose that upon performing suitable local measurements on a state $\rho^P$ the prover is able to produce correlations reaching the value $\beta \leq \beta_q$ for the considered Bell expression $I_N$, chosen such that it allows to self-test a subspace $\mathcal{H}_s = \text{span} (\{\ket{\psi_i}\}_{i=1}^k)$.
In order to formalize the problem in a quantitative manner we employ
the \textit{subspace extractability} which is a generalization of
the state extractability introduced in Ref. \cite{Bardyn2009} and used in 
the self-testing context in Ref. \cite{Jed2} to the case of subspaces. 
It is defined as
\begin{equation}\label{extract}
\Theta (\rho^P \rightarrow \mathcal{H}_s ) = \max_{\sigma_s \in \mathcal{B}(\mathcal{H}_s)} \max_{\Lambda_{i}} F(\sigma_s ,(\Lambda_1 \otimes \ldots \otimes\Lambda_N) (\rho^P)), 
\end{equation}
where the first maximum is over all mixed states supported on the considered subspace $\mathcal{H}_s$, whereas the second maximum is over 
all local quantum channels $\Lambda:\mathcal{B}(\mathcal{H}_{P_i})\to \mathcal{B}(\mathbbm{C}^2)$. Finally,
$F(\rho,\sigma)$ stands for the state fidelity defined as
\begin{equation}
F(\rho,\sigma)=\|\sqrt{\rho}\sqrt{\sigma}\|_1^2
\end{equation}
with $\|\cdot\|_1$ being the trace norm, $\|X\|_1=\mathrm{tr}\sqrt{X^{\dagger}X}$. 
Notice that here, for the sake of consistency with previous approaches, we formulate our self-testing statement in terms of local channels instead of local unitaries. The two formulations are equivalent, since one can simply trace out the purification space $E$ in the l.h.s of Eq. (3) in the main text and obtain local channels acting on $\rho^P$. Hence, the self-testing properties of the inequality $I_N$ imply that, if the observed violation is $\beta = \beta_q$, then $\Theta (\rho^P \rightarrow \mathcal{H}_s ) = 1$.

Here we show how to derive non-trivial statements for cases where  $\beta < \beta_q$, namely by bounding the subspace extractabiliy as a function of the observed violation. To do so, we make use of the following Fact, proven at the end of the section:

\begin{fakt}\label{extra}
Consider a Hilbert space $\mathcal{H}$ and a subspace $\mathcal{H}_s\subseteq\mathcal{H}$. Then, for any $\rho$ acting on $\mathcal{H}$
the following identity holds true
\begin{eqnarray}\label{dupa}
\max_{\sigma\in \mathcal{B}(\mathcal{H}_{s}) } F(\sigma ,\rho)=\mathrm{tr} (P_s \rho),
\end{eqnarray}
where $P_{s}$ is the projector onto the subspace $\mathcal{H}_s$ and the maximum is taken over all density matrices supported on $\mathcal{H}_s$.
\end{fakt}

The above fact allows us to simplify the expression of the extractability to
\begin{equation}\label{extract}
\Theta (\rho_P \rightarrow \mathcal{H}_s ) = \, \, \max_{\Lambda_{i}} \, \, \text{tr} [P_s (\Lambda_1 \otimes \ldots \otimes\Lambda_N) (\rho_P)], 
\end{equation}
Notice that the maximum over states supported on $\mathcal{H}_s$ in Eq. (\ref{extract}) is obtained by choosing $\sigma_s = (P_s \rho_{\Lambda} P_s)/ \text{tr} (P_s \rho_{\Lambda})$, where $\rho_{\Lambda} =( \Lambda_1 \otimes \ldots \otimes\Lambda_N) (\rho_P)$.

The above bound leads us to define a natural generalisation of the approach in \cite{Jed2} to the case of subspace self-testing.
Indeed, note that the lower bound on the fidelity in Eq. (\ref{extract})
can equivalently be written as 
\begin{equation}
\text{tr} [\rho_P (\Lambda_1^{\dagger}\otimes\ldots\otimes\Lambda_N^{\dagger})(P_s)] \, ,
\end{equation} 
where $\Lambda_i^{\dagger}$ are dual maps of the quantum channels $\Lambda_i$. Now, proving for some particular channels $\Lambda_i$  
an operator inequality 
\begin{equation}\label{OpIneq}
K:=(\Lambda_1^{\dagger}\otimes\ldots\otimes\Lambda_N^{\dagger})(P_s)\geq a \mathcal{B}_N+ b \mathbbm{1}
\end{equation}
with for some $a,b \in\mathbbm{R}$, where $\mathcal{B}_N$
stands for the Bell operator corresponding to 
the Bell inequality $I_N$ and constructed 
from any possible dichotomic observables, would imply the following inequality for the subspace extractability
\begin{equation}
\Theta(\rho_P \to \mathcal{H}_s)\geq a \beta+ b.
\end{equation}
Proving an operator inequality (\ref{OpIneq})
for arbitrary local observables is generally very challenging. However, due to the fact that here we consider the simplest Bell scenario involving
two dichotomic measurements per site, one can exploit 
Jordan's lemma, which, analogously to what is done in Ref. \cite{Jed2}, allows us to reduce the problem to an $N$-qubit space. That is, 
the local observables $A_{x_i}^{(i)}$ appearing in $\mathcal{B}_N$ can now be parametrized as
\begin{equation}\label{eq:meas}
A_{x_1}^{(1)} = \cos{\alpha_1}\, X + 
(-1)^{x_1} \sin{\alpha_1}\, Z, 
\end{equation}
and
\begin{equation}\label{eq:measi}
A_{x_i}^{(i)} = \cos{\alpha_i}\, H + 
(-1)^{x_i} \sin{\alpha_i}\, V  
\end{equation}
for $i = 2, \ldots, N$,
where $H = (X + Z)/\sqrt{2}$, $V = (X - Z)/\sqrt{2}$ and $\alpha_i \in [ 0,\pi/2] $. This gives rise to 
a Bell operator $\mathcal{B}_{N}(\vec{\alpha})$ that now depends on the angles $\alpha_i$. We then consider
particular quantum channels
\begin{equation}
\Lambda_i(x) [\rho] = \frac{1 + g(x)}{2} \rho + 
\frac{1 - g(x)}{2} \Gamma_i (x) \rho \Gamma_i (x)
\end{equation}
where the dependence on the measurement angle is encoded in the function $g(x) = (1+ \sqrt{2} )(\sin{x}+\cos{x} -1)$ together with the definition of $\Gamma(x) = M_i^{a}$ if $x \leq \pi /4$ and $\Gamma(x) = M_i^{b}$ if $x > \pi /4$. Lastly, we define $M_1^{a,b} = X,
Z$ and  $M_i^{a,b} = H,V$ for $i = 2, \ldots, N$.

It is now enough to prove that for all possible
choices of $\alpha_i$, the following inequality is
satisfied
\begin{equation}\label{eq:fidbound}
K(\alpha_1,\ldots,\alpha_N)\geq a\mathcal{B}_N(\alpha_1,\ldots,\alpha_N)+b \mathbbm{1}
\end{equation}
for some choice of $a,b \in\mathbbm{R}$. 

\begin{figure}
\includegraphics[width=0.45\textwidth]{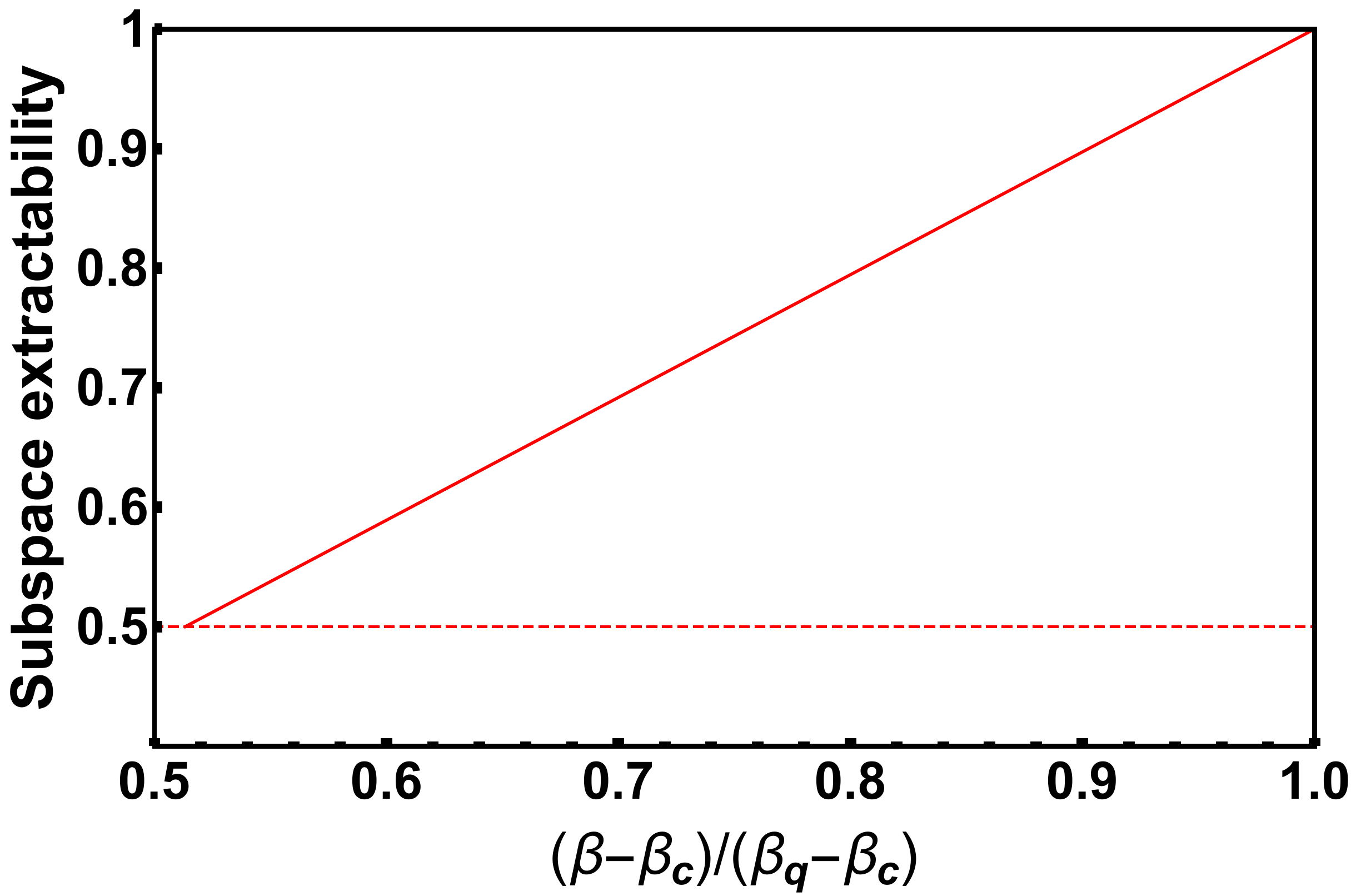}
\caption{Numerical estimation of the minimal subspace extractability for the target 5-qubit Shor's code subspace as a function of the relative observed violation $(\beta -  \beta_c)/(\beta_q - \beta_c)$ of the corresponding Bell inequality $I_5$.}
\label{fig:fidelity SM}
\end{figure}

In order to give an indication of the noise robustness of our subspace self-testing protocols, we have performed numerical tests to obtain indicative bounds \eqref{eq:fidbound} for the inequality $I_5$ in Eq. (5) in the main text. Recall that such inequality is maximally violated by any state belonging to the five-qubit error correction code subspace. The fact that such error correction code is a stabilizer code allows us conveniently represent the projector onto 
the stabilized subspace as $P_5 = \prod_{i = 1} ^4 (\mathbbm{1} + S_i) / 2$, where the four operators $S_i$ are the stabilizing operators defined in Eq. (4) of the main text.
The applied numerical procedure works as follows: given a fixed $a$, estimate the corresponding $b$ by numerically computing the minimal eigenvalue of the operator $K - a \mathcal{B}_G$ and minimizing over all the angles $\alpha_i$.
The above-mentioned minimization is performed in two steps: for a fixed value of the angles, we compute the minimal eigenvalue $b (\alpha_1,\ldots \alpha_N )$ of the $K - a \mathcal{B}_G$ operator by means of the \textsf{eig} function in MATLAB.
The second steps is then to find the configuration of angles that leads to the minimal value of $b (\alpha_1,\ldots \alpha_N )$. Notice that this second optimization is generally non-convex and indeed hard to perform in general. Nevertheless, since the optimization can be run over a compact parameter space (each measurement angle can vary from $0 $ to $2 \pi$), one can adapt the method put forward in \cite{coopmans2019robust} and discretize such intervals in $100$ steps, computed the minimal eigenvalue $b$ of the operator $K - b \mathcal{B}_G$ for each of them and identify for which value of the angles one obtains the lowest value of $b$. 

Notice that to have a fidelity bound that leads to fidelity $1$ at the point of maximal violation, the inequality \eqref{eq:fidbound} has to become tight for the measurements angles leading to the maximal violation, that is $\alpha_i = \pi/4$ for all $i = 1, \ldots, N$.
In order to get the best slope for our numerical results, we explored several values of $a$ by proceeding at paces of $0.001$ and estimated numerically the minimum value of $a$ for which the corresponding conjectured bound still satisfied such a property. This led to the numerical results presented in Fig. \ref{fig:fidelity SM}.

Let us now conclude by proving Fact \ref{extra}.
\begin{proof}
Let us consider a state $\sigma$ supported on $\mathcal{H}_s$. One has the following identities
\begin{eqnarray}\label{culo}
F(\sigma,\rho)&=&\|\sqrt{\sigma}\sqrt{\rho}\|_1^2\nonumber\\
&=&\|\sqrt{\sigma}\sqrt{P_{s}\rho P_{s}}\|_1^2\nonumber\\
&=&\mathrm{tr}(P_{s}\rho)\|\sqrt{\sigma}\sqrt{\rho_s }\|_1^2\nonumber\\
&=&\mathrm{tr}(P_{s}\rho) F(\sigma,\rho_s),
\end{eqnarray}
where in the third line we multiplied and divided by $\mathrm{tr}(P_{s}\rho)$, and defined  $\rho_s=P_{s}\rho P_{s}/\mathrm{tr}(P_{s}\rho)$, while the second equality follows from the fact that 
\begin{eqnarray}
\|\sqrt{\sigma}\sqrt{\rho}\|_1&=&\mathrm{tr}\sqrt{\sqrt{\sigma}\rho\sqrt{\sigma}}
\nonumber\\
&=&\mathrm{tr}\sqrt{\sqrt{\sigma}P_{s}\rho P_{s}\sqrt{\sigma}}\nonumber\\
&=&\|\sqrt{\sigma}\sqrt{P_{s}\rho P_{s}}\|_1,
\end{eqnarray}
which is due to the fact that $P_{s}\rho P_{s}\geq 0$ and that $\sigma P_{s}=P_{\sigma}\sigma=\sigma$, which implies
$\sqrt{\sigma} P_{s}=P_{s}\sqrt{\sigma}=\sqrt{\sigma}$.

Now, the identity (\ref{culo}) allows us to write
\begin{equation}
\max_{\sigma\in\mathcal{H}_s}F(\sigma,\rho)=\mathrm{tr}(P_s\rho)\max_{\sigma\in \mathcal{H}_s} F(\sigma,\rho_s).
\end{equation}
Since $\rho_s$ is supported on $\mathcal{H}_s$, one clearly finds that $\max_{\sigma\in \mathcal{H}_s} F(\sigma,\rho_s)=1$, which leads to (\ref{dupa}), completing the proof.
\end{proof}

\section{APPENDIX D: GEOMETRICAL STRUCTURE OF THE SET OF POINTS MAXIMALLY VIOLATING $I_N^{tor}$}

Here we show that the correlations maximally violating 
the Bell inequalities introduced in the main text span an
affine spaces in the boundary of the set of quantum correlations of dimension higher than zero. In particular, this dimension equals to one for the case of the five-qubit code, while it is at least two for the case of the toric code of any lattice size $L$.

We start by analyzing the simpler case of the five-qubit code.
First, we identify two orthogonal states $\ket{\psi_{1,2}}$ in the stabilized subspace $C_5$ as those associated to the eigenvalues $-1,+1$ for the operator $S_5 = \sz^{(1)} \sz^{(2)} \sz^{(3)} \sz^{(4)}  \sz^{(5)}$. This can be done because $S_5$ commutes with all the generators (\ref{5qubitStabII}) and it is independent from them.
By performing the local measurements leading to the maximal Bell violation of $I_5$ on these two states, the resulting behaviours $\mathcal{P}_1$, $\mathcal{P}_2$ are different. To see this, we apply to $S_5$ the map between Pauli matrices and observables that we use to derive the Bell inequalities \eqref{eq:I5} and \eqref{eqapp:Itor}. By using the fact that $\ket{\psi_{1,2}}$ are eigenstates of that operators, one derives that the corresponding behaviours lead to expectation values satisfying
\begin{equation}
\langle \tilde{S}_5 \rangle_{\mathcal{P}_i} = \frac{1}{\sqrt{2}}\langle (A_0^{(1)} - A_1^{(1)}) A_1^{(2)}  A_1^{(3)}   A_1^{(4)}   A_1^{(5)} \rangle_{\mathcal{P}_i} = (-1)^i,    
\end{equation}
which can only be fulfilled if the two correlators involved take different values for $\mathcal{P}_1$ and $\mathcal{P}_2$.
Since the two behaviours take the same values for the expectation values appearing in the Bell expression $I_5$, it follows that the corresponding correlation vectors
$\mathcal{P}_1$ and $\mathcal{P}_2$ are linearly independent, hence spanning an affine subspace of dimension one.

Let us now move to the case of the toric code, by following a similar reasoning as above. We introduce the following operators
\begin{equation}
    \mathcal{Z}_{\text{hor}}=\prod_{i=0}^{L-1} \sz^{(j+i)}
\end{equation}
and
\begin{equation}
    \mathcal{Z}_{\text{vert}}=\prod_{i=1}^{L-1} \sz^{(j + 2i)} \quad .
\end{equation}
They consist of a product of $\sz$ operators acting respectively on an horizontal and vertical loop around the torus, containing a reference qubit $j$. It is direct to see that they mutually commute and that they
commute with every plaquette $S_p$ and vertex $S_v$ operator. 
Moreover, they are independent of the stabilizer
$\mathbb{S}_N^{\mathrm{tor}}$. Hence, we can use them to define an orthonormal basis for the two-qubit subspace of the toric code.
More precisely, such a basis can be chosen to be the collection of four states $\lbrace \ket{\psi_{ab}} \rbrace_{a,b = \pm 1}$ in $C_N^{\mathrm{tor}}$ defined as the eigenstates of $\mathcal{Z}_{\text{hor}}$ ($\mathcal{Z}_{\text{ver}}$) with eigenvalue $a$ ($b$).

By performing the measurements leading to the maximal violation of \eqref{eqapp:Itor} on each of those states, one obtains the corresponding correlation points $\mathcal{P}_{ab}$.
To show that, for $a,b = \pm 1$, these four points are represented by linearly independent vectors, it is enough to make use of the fact that the related states $\ket{\psi_{ab}}$ are eigenvectors of the two additional operators  $\mathcal{Z}_{\text{hor}}$ and $\mathcal{Z}_{\text{vert}}$.
In particular, let us assume that the reference qubit $j$ is the one where the substitution $\sx^{(j)},\sz^{(j)} \rightarrow (\rA^{(j)}_0 \pm \rA^{(j)}_1)/\sqrt{2}$ has been made.
Then if follows that the corresponding correlation points must satisfy 
\begin{equation}
\langle \tilde{\mathcal{Z}}_{\text{hor}} \rangle_{\mathcal{P}_{ab}} = \frac{1}{\sqrt{2}} \left\langle ( \rA^{(j)}_0 - \rA^{(j)}_1) \prod_{i = 1}^{L-1} \rA^{(j+i)}_1 \right\rangle_{\mathcal{P}_{ab}} = a      
\end{equation}
and
\begin{equation}
\langle \tilde{\mathcal{Z}}_{\text{vert}} \rangle_{\mathcal{P}_{ab}} = \frac{1}{\sqrt{2}} \left\langle ( \rA^{(j)}_0 - \rA^{(j)}_1) \prod_{i = 1}^{L-1} \rA^{(j+2i)}_1 \right\rangle_{\mathcal{P}_{ab}} = b \quad .
\end{equation}
Let us notice that one can follow the same reasoning for every other pair of vertical and horizontal loop operators, leading to similar conditions for other $L-$body expectation values arising from the $\mathcal{P}_{ab}$'s. By combining the above conditions with the fact that the four behaviours take the same values on the expectation values involved in $I^{tor}_N$, one can easily convince themself that at least three of the four corresponding correlation vectors are linearly independent for any $N$. Hence, the affine subspace spanned by the four behaviours has dimension at least two.

\bibliography{Subspaces}